\documentclass[a4paper,UKenglish]{lipics-v2018}
\pdfoutput=1

\usepackage{microtype}
\usepackage
[
	algo2e,
	ruled,
	vlined,
	linesnumbered
]
{algorithm2e}
\usepackage{mathtools}
\DeclareMathOperator{\level}{level}
\usepackage{booktabs}
\usepackage[table]{xcolor}
\usepackage{siunitx}

\usepackage{xspace}

\usepackage{timetravel}
\setCounters{theorem}

\usepackage[numbers,sort&compress]{natbib}
\usepackage{todonotes}
\usepackage{paralist}

\title{Efficiently Generating Geometric Inhomogeneous and Hyperbolic
  Random Graphs}

\titlerunning{Efficiently Generating Geometric Inhomogeneous and Hyperbolic
  Random Graphs}

\author{Thomas Bl\"asius}{Hasso Plattner Institute, Potsdam, Germany}{thomas.blaesius@hpi.de}{}{}
\author{Tobias Friedrich}{Hasso Plattner Institute, Potsdam, Germany}{tobias.friedrich@hpi.de}{}{}
\author{Maximilian Katzmann}{Hasso Plattner Institute, Potsdam, Germany}{maximilian.katzmann@hpi.de}{}{}
\author{Ulrich Meyer}{Goethe University, Frankfurt, Germany}{umeyer@ae.cs.uni-frankfurt.de}{}{}
\author{Manuel Penschuck}{Goethe University, Frankfurt, Germany}{mpenschuck@ae.cs.uni-frankfurt.de}{}{}
\author{Christopher Weyand}{Hasso Plattner Institute, Potsdam, Germany}{christopher.weyand@hpi.de}{}{}

\authorrunning{T.~Bl\"asius, T.~Friedrich, M.~Katzmann, U.~Meyer, M.~Penschuck, C.~Weyand}

\Copyright{Anonymous}

\category{}

\relatedversion{}

\supplement{}

\funding{}

\EventEditors{John Q. Open and Joan R. Access}
\EventNoEds{2}
\EventLongTitle{42nd Conference on Very Important Topics (CVIT 2016)}
\EventShortTitle{CVIT 2016}
\EventAcronym{CVIT}
\EventYear{2016}
\EventDate{December 24--27, 2016}
\EventLocation{Little Whinging, United Kingdom}
\EventLogo{}
\SeriesVolume{42}
\ArticleNo{23}
\nolinenumbers 
\hideLIPIcs  

\DeclareMathOperator{\acosh}{\ensuremath{arccosh}}

\newcommand{\formatImpl}[1]{\textsc{#1}}
\newcommand{\nkquad}{\formatImpl{NkQuad}\xspace}
\newcommand{\nkgen}{\formatImpl{NkGen}\xspace}
\newcommand{\nkopt}{\formatImpl{NkOpt}\xspace}
\newcommand{\hypergen}{\formatImpl{HyperGen}\xspace}
\newcommand{\rhg}{\formatImpl{RHG}\xspace}
\newcommand{\srhg}{\formatImpl{sRHG}\xspace}
\newcommand{\hypergirgs}{\formatImpl{HyperGIRGs}\xspace}
\newcommand{\embedder}{\formatImpl{Embedder}\xspace}

\newcommand{\pairwise}{\formatImpl{Pairwise}\xspace}
\newcommand{\quadtree}{\formatImpl{QuadTree}\xspace}

\newcommand{\dhrg}{\ensuremath{d_{\mathrm{HRG}}}}
\newcommand{\dgirg}{\ensuremath{d_{\mathrm{GIRG}}}}
\newcommand{\Dgirg}{\ensuremath{D_{\mathrm{GIRG}}}}

\begin{document}

\maketitle

\begin{abstract}
  Hyperbolic random graphs (HRG) and geometric inhomogeneous random
  graphs (GIRG) are two similar generative network models that were
  designed to resemble complex real world networks.  In particular,
  they have a power-law degree distribution with controllable
  exponent $\beta$, and high clustering that can be
  controlled via the temperature $T$.

  We present the first implementation of an efficient GIRG generator
  running in expected linear time.  Besides varying temperatures, it
  also supports underlying geometries of higher dimensions.  It is
  capable of generating graphs with ten million edges in under a
  second on commodity hardware.  The algorithm can be adapted to HRGs.
  Our resulting implementation is the fastest sequential HRG
  generator, despite the fact that we support non-zero temperatures.
  Though non-zero temperatures are crucial for many applications, most
  existing generators are restricted to $T = 0$.  
  Our generators support parallelization, although this is not the focus of this paper.
  We note that our generators draw from the correct
  probability distribution, i.e., they involve no approximation.

  Besides the generators themselves, we also provide an efficient
  algorithm to determine the non-trivial dependency between the
  average degree of the resulting graph and the input parameters of
  the GIRG model.  This makes it possible to specify the expected average degree as
  input.

  Moreover, we investigate the differences between HRGs and GIRGs,
  shedding new light on the nature of the relation between the two
  models.  Although HRGs represent, in a certain sense, a special case
  of the GIRG model, we find that a straight-forward inclusion does
  not hold in practice. However, the difference is negligible for most
  use cases.
\end{abstract}


\newpage

\section{Introduction}

Network models play an important role in different fields of science.
From the perspective of network science, models can be used to explain
observed behavior in the real world.  To mention one example, Watts
and Strogatz~\cite{ws-cdswn-98} observed that few random long-range
connections suffice to guarantee a small diameter.  This explains why
many real-world networks exhibit the small-world property despite
heavily favoring local over long-range connections.  From the
perspective of computer science, and specifically algorithmics,
realistic random networks can provide input instances for graph
algorithms.  This facilitates theoretical approaches (e.g.,
average-case analysis), as well as extensive empirical evaluations by
providing an abundance of benchmark instances, solving the pervasive
scarcity of real-world instances.

There are some crucial features that make a network model useful.  The
generated instances have to resemble real-world networks.  The model
should be as simple and natural as possible to facilitate theoretical
analysis, and to prevent untypical artifacts.  And it must be possible
to efficiently draw networks from the model.  This is particularly
important for the empirical analysis of model properties and for
generating benchmark instances.

A model that has proven itself useful in recent years is the
\emph{hyperbolic random graph (HRG)} model~\cite{kpk-hgcn-10}.  HRGs
are generated by drawing vertex positions uniformly at random from a
disk in the hyperbolic plane.  Two vertices are joined by an edge if
and only if their distance lies below a certain threshold; see
Section~\ref{sec:hrgmodel}.  HRGs resemble real-world networks with
respect to crucial properties.  Most notable are the \emph{power-law
  degree distribution}~\cite{gpp-rhg-12} (i.e., the number of vertices
of degree $k$ is roughly proportional to $k^{-\beta}$ with
$\beta \in (2, 3)$), the high \emph{clustering
  coefficient}~\cite{gpp-rhg-12} (i.e., two vertices are more likely
to be connected if they have a common neighbor), and the small
diameter~\cite{fk-dhrg-18,ms-dkrg-17}.  Moreover, HRGs are accessible
for theoretical analysis (see,
e.g.,~\cite{gpp-rhg-12,fk-dhrg-18,ms-dkrg-17,bff-espsf-18}).  Finally
there is a multitude of efficient generators with different
emphases~\cite{aok-hgg-15,lmp-grhgst-15,lm-qpnsdse-16,lolm-gmcnh-16,p-gprhg-17,fls-cfmdgg-18,flm-cfmdgg-17};
see Section~\ref{sec:comp-with-exist} for a discussion.

Closely related to HRGs is the \emph{geometric inhomogeneous random
  graph (GIRG)} model~\cite{bkl-girg-19}.  Here every vertex has a
position on the $d$-dimensional torus and a weight following a power
law.  Two vertices are then connected if and only if their distance on
the torus is smaller than a threshold based on the product of their
weights.  When using positions on the circle ($d = 1$), GIRGs
approximate HRGs in the following sense: the processes of generating a
HRG and a GIRG can be coupled such that it suffices to decrease and
increase the average degree of the GIRG by only a constant factor to
obtain a subgraph and a supergraph of the corresponding HRG,
respectively.
Compared to HRGs, GIRGs are potentially easier to analyze, generalize
nicely to higher dimensions, and the weights allow to directly adjust
the degree distribution.

Above, we described the idealized \emph{threshold variants} of the
models, where two vertices are connected if an only if their distance
is small enough.  Arguably more realistic are the \emph{binomial
  variants}, which allow longer edges and shorter non-edges with a
small probability.  This is achieved with an additional parameter $T$,
called \emph{temperature}.  For $T \to 0$, the binomial and threshold
variants coincide.  Many publications focus on the threshold case, as
it is typically simpler.  This is particularly true for generation
algorithms: in the threshold variants one can ignore all vertex pairs
with sufficient distance, which can be done using geometric data
structures.  In the binomial case, any pair of vertices could be
adjacent, and the search space cannot be reduced as easily.  For
practical purposes, however, a non-zero temperature is crucial as
real-world networks are generally assumed to have positive temperature.
Moreover, from an algorithmic perspective, the threshold variants
typically produce particularly well-behaved instances, while a higher
temperature leads to difficult problem inputs.  Thus, to obtain
benchmark instances of varying difficulty, generators for the binomial
variants are key.
    
\subsection{Contribution \& Outline}
\label{sec:contribution}

Based on the algorithm by Bringmann, Keusch, and
Lengler~\cite{bkl-girg-19}, we provide an efficient and flexible GIRG
generator.  It includes the binomial case and allows higher
dimensions.  Its expected running time is linear.  To the best of our
knowledge, this is the first efficient generator for the GIRG model.
Moreover, we adapt the algorithm to the HRG model, including the
binomial variant.  Compared to existing HRG generators (most of which
only support the threshold variant), our implementation is the fastest
sequential HRG generator.

A refactoring of the original GIRG algorithm~\cite{bkl-girg-19} allows
us to parallelize our generators.  They do not use multiple processors
as effectively as the threshold-HRG generator by
Penschuck~\cite{p-gprhg-17}, which was specifically tailored towards
parallelism.  However, in a setting realistic for commodity hardware
(8 cores, 16 threads), we still achieve comparable run times.

Our generators come as an open-source C++ library\footnote{\url{https://github.com/chistopher/girgs}}
with documentation, command-line interface, unit tests, micro
benchmarks, and OpenMP~\cite{openmp-oapiv-18} parallelization using
shared memory.  An integration into NetworKit~\cite{ssm-n-16} is
planned.

Besides the efficient generators, we have three secondary
contributions.
\begin{inparaenum}[(I)]
\item We provide a comprehensible description of the sampling
  algorithm that should make it easy to understand how the algorithm
  works, why it works, and how it can be implemented.  Although the
  core idea of the algorithm is not new~\cite{bkl-girg-19}, the
  previous description is somewhat technical.
\item The expected average degree can be controlled via an input
  parameter.  However, the dependence of the average degree on the
  actual parameter is non-trivial.  In fact, given the average degree,
  there is no closed formula to determine the parameter.  We provide a
  linear-time algorithm to estimate it.
\item We investigate how GIRGs and HRGs actually relate to each other
  by measuring how much the average degree of the GIRG has to be
  decreased and increased to obtain a subgraph and supergraph of the
  HRG, respectively.  
\end{inparaenum}

In the following we first discuss our main contribution in the context
of existing HRG generators.  In Section~\ref{sec:models}, we formally
define the GIRG and HRG models.  Afterwards we describe the sampling
algorithm in Section~\ref{sec:algorithm}.  In Section~\ref{sec:impl}
we discuss implementation details, including the parameter estimation
for the average degree (Section~\ref{sec:constants}) as well as
multiple performance improvements.  Section~\ref{sec:experiments}
contains our experiments: we investigate the scaling behavior of our
generator in Section~\ref{sec:girgexperiments}, compare our HRG
generator to existing ones in Section~\ref{sec:hrgexperiments}, and
compare GIRGs to HRGs in Section~\ref{sec:hrgisgirgexperiments}.

\subsection{Comparison with Existing Generators}
\label{sec:comp-with-exist}

\begin{table}
  \definecolor{tablegray}{gray}{0.9}
  \rowcolors{1}{white}{tablegray}
  \centering
  \caption{Existing hyperbolic random graph generators.  The columns
    show the names used throughout the paper; the conference
    appearance; a reference (journal if available); whether the
    generator supports the binomial model; and the asymptotic running
    time.  The time bounds hold in the worst case (wc), with high
    probability (whp), in expectation (exp), or empirically~(emp).}
  \begin{tabular}{l l  c c   c   c   c}
    \toprule
    Name           & First Published & Ref.                   & Binom.     & Running Time                   \\
    \midrule
    \pairwise      & CPC'15          & \cite{aok-hgg-15}      & \checkmark & $\Theta(n^2)$ (wc)             \\
    \quadtree      & ISAAC'15        & \cite{lmp-grhgst-15}   &            & $O((n^{3/2} + m) \log n)$ (wc) \\
    \nkquad        & IWOCA'16        & \cite{lm-qpnsdse-16}   & \checkmark & $O((n^{3/2} + m) \log n)$ (wc) \\
    \nkgen, \nkopt & HPEC'16         & \cite{lolm-gmcnh-16}   &            & $O(n\log n+m)$ (emp)           \\
    \embedder      & ESA'16          & \cite{bfkl-eesfghp-18} & \checkmark & $\Theta(n+m)$ (exp)            \\
    \hypergen      & SEA'17          & \cite{p-gprhg-17}      &            & $O(n\log\log n + m)$ (whp)     \\
    \rhg           & IPDPS'18        & \cite{fls-cfmdgg-18}   &            & $\Theta(n+m)$ (exp)            \\
    \srhg          & preprint        & \cite{flm-cfmdgg-17}   &            & $\Theta(n+m)$ (exp)            \\
    \hypergirgs    & this paper      &                        & \checkmark & $\Theta(n+m)$ (exp)            \\
    \bottomrule
  \end{tabular}
  \label{tab:related}
\end{table}

We are not aware of a previous GIRG generator.  Concerning HRGs, most
previous algorithms only support the threshold case; see
Table~\ref{tab:related}.  The only published exceptions are the
trivial quadratic algorithm~\cite{aok-hgg-15}, and an
$O((n^{3/2} + m) \log n)$ algorithm~\cite{lm-qpnsdse-16} based on a
quad-tree data structure~\cite{lmp-grhgst-15}.  The latter is part of
NetworkKit; we call it \nkquad.  Moreover, the code for a hyperbolic
embedding algorithm~\cite{bfkl-eesfghp-18} includes an HRG generator
implemented by Bringmann based on the GIRG
algorithm~\cite{bkl-girg-19}; we call it \embedder in the following.
\embedder has been widely ignored as a high performance generator.
This is because it was somewhat hidden, and it is heavily outperformed
by other threshold generators.  Experiments show that our generator
\hypergirgs is much faster than \nkquad, which is to be expected
considering the asymptotic running time.  Moreover, on a single
processor, we outperform \embedder by an order of magnitude for
$T = 0$ and by a factor of~$4$ for higher temperatures.  As \embedder
does not support parallelization, this speed-up increases for multiple
processors.

For the threshold variant of HRGs, there are the following generators.
The quad-tree data structure mentioned above was initially used for
a threshold generator (\quadtree)~\cite{lmp-grhgst-15}.  It was later
improved leading to the algorithm currently implemented in NetworKit
(\nkgen)~\cite{lolm-gmcnh-16}.  A later re-implementation by
Penschuck~\cite{p-gprhg-17} improves it by about a factor of~2
(\nkopt).  However, the main contribution of
Penschuck~\cite{p-gprhg-17} was a new generator that features
sublinear memory and near optimal parallelization (\hypergen).  Up to
date, \hypergen was the fastest threshold-HRG generator on a single
processor.  Our generator, \hypergirgs, improves by a factor of~$1.3$
-- $2$ (depending on the parameters) but scales worse for more
processors.  Finally, Funke et al.~\cite{fls-cfmdgg-18} provide a
generator designed for a distributed setting to generate enormous
instances (\rhg).  Its run time was later further reduced
(\srhg)~\cite{flm-cfmdgg-17}.


\section{Models}
\label{sec:models}

\subsection{Geometric Inhomogeneous Random Graphs}
\label{sec:girgmodel}

GIRGs~\cite{bkl-girg-19} combine elements from random geometric
graphs~\cite{g-rpn-61} and Chung-Lu
graphs~\cite{cl-adrgged-02,cl-ccrgg-02}.  
%
Let $V = \{1,\dots,n\}$ be a set of vertices with weights
$w_1,\dots,w_n$ following a power law with exponent $\beta>2$.  Let
$W$ be their sum.  Let $\mathbb{T}^d$ be the $d$-dimensional torus for
a fixed dimension $d\geq1$ represented by the $d$-dimensional cube
$[0,1]^d$ where opposite boundaries are identified.  For each vertex
$v \in V$, let $x_v\in\mathbb{T}^d$ be a point drawn uniformly and
independently at random.  For $x,y \in \mathbb{T}^d$ let $||x-y||$
denote the $L_\infty$-norm on the torus, i.e.  $||x-y|| = \max_{1\leq
  i\leq d} \min\{|x_i-y_i|, 1-|x_i-y_i|\}$.  Two vertices $u\neq v$
are independently connected with probability $p_{uv}$.  For a positive
temperature $0 < T < 1$,
\begin{equation}
  \label{eq:girg-binom}
p_{uv} = \min\left\{ 1, c\left( \frac{w_uw_v/W}{{||x_u - x_v||}^d} \right)^{1/T} \right\}
\end{equation}
while for $T=0$ a threshold variant of the model is obtained with
\begin{equation*}
  \label{eq:girg-thresh}
p_{uv} = 
\begin{cases}
	1 & \text{if } ||x_u - x_v|| \leq c(w_uw_v/W)^{1/d}, \\
	0 & \text{else.}
\end{cases}
\end{equation*}
The constant $c>0$ controls the expected average degree.  We note that
the above formulation slightly deviates from the original definition;
see Section~\ref{sec:comp-girgs-hrgs} for more details.


\subsection{Hyperbolic Random Graphs}
\label{sec:hrgmodel}

HRGs~\cite{kpk-hgcn-10} are generated by sampling random positions in
the hyperbolic plane and connecting vertices that are close.  More
formally, let $V = \{1,\dots,n\}$ be a set of vertices.  Let
$\alpha > 1/2$ and $C \in \mathbb{R}$ be two constants, where $\alpha$
controls the power-law degree distribution with exponent
$\beta = 2\alpha{+}1 > 2$, and $C$ determines the average degree
$\bar d$.  For each vertex $v \in V$, we sample a random point
$p_v = (r_v, \theta_v)$ in the hyperbolic plane, using polar
coordinates.  Its angular coordinate $\theta_v$ is chosen uniformly
from $[0, 2\pi]$ while its radius $0 \le r_v < R$ with
$R = 2\log(n) + C$ is drawn according to the density function
$f(r) = \frac{\alpha \sinh(\alpha r)}{\cosh(\alpha R) - 1}$.  In the
threshold case of HRGs two vertices $u \not= v$ are connected if and
only if their distance is below~$R$.  The hyperbolic distance
$d(p_u, p_v)$ is defined via
$\cosh( d(p_u, p_v) ) = \cosh(r_u)\cosh(r_v) -
\sinh(r_u)\sinh(r_v)\cos(\theta_u - \theta_v)$, where the angle
difference $\theta_u - \theta_v$ is modulo~$\pi$.

The binomial variant adds a temperature~$T \in [0, 1]$ to control the
clustering, with lower temperatures leading to higher clustering.  Two
nodes $u,v \in V$ are then connected with probability
$p_T(d(p_u, p_v))$ where $p_T(d) = (\exp[ (d-R)/(2T) ] + 1)^{-1}$.
For $T \rightarrow 0$, the two definitions (threshold and binomial)
coincide.

\subsection{Comparison of GIRGs and HRGs}
\label{sec:comp-girgs-hrgs}

Bringmann et al.~\cite{bkl-girg-19} show that the HRG model can be
seen as a special case of the GIRG model in the following sense.  Let
$\dhrg$ be the average degree of a HRG.  Then there exist GIRGs with
average degree $\dgirg$ and $\Dgirg$ with $\dgirg \le \dhrg \le
\Dgirg$ such that they are sub- and supergraphs of the HRG,
respectively.  Moreover, $\dgirg$ and $\Dgirg$ differ only by a
constant factor.  Formally, this is achieved by using the big-O
notation instead of a single constant $c$ for the connection
probability.  We call this the \emph{generic GIRG framework}.  It
basically captures any specific model whose connection probabilities
differ from Equation~(\ref{eq:girg-binom}) by only a constant factor.
From a theoretical point of view this is useful as proving something
for the generic GIRG framework also proves it for any manifestation,
including HRGs.

To see how HRGs fit into the generic GIRG framework, consider the
following mapping~\cite{bkl-girg-19}.  Radii are mapped to weights
$w_v = e^{(R-r_v)/2}$, 
and angles are scaled to fit on a $1$-dimensional torus
$x_v = \theta_v / (2\pi)$.  One can then see that the hyperbolic
connection probability $p_T(d)$ under the provided mapping deviates
from Equation~(\ref{eq:girg-binom}) by only a constant.  Thus, $c$ in
Equation~(\ref{eq:girg-binom}) can be chosen such that all GIRG
probabilities are larger or smaller than the corresponding HRG
probabilities, leading to the two average degrees $\dgirg$ and
$\Dgirg$ mentioned above.
Bringmann et al.~\cite{bkl-girg-19} note that the two constants, which
they hide in the big-$O$ notation, do not have to match.  They leave
it open if they match, converge asymptotically, or how large the
interval between them is in practice.  We investigate this empirically
in Section~\ref{sec:hrgisgirgexperiments}.


\newcommand{\maxLevel}{D}
\newcommand{\maxLayer}{K}
\newcommand{\comparisonLevel}[2]{CL(#1, #2)}
\newcommand{\insertionLevel}[1]{I(#1)}

\section{Sampling Algorithm}
\label{sec:algorithm}

As mentioned in the introduction, the core of our sampling algorithm
is based on the algorithm by Bringmann et al.~\cite{bkl-girg-19}.  In
the following, we first give a description of the core ideas and then
work out the details that lead to an efficient implementation.

To explain the idea, we make two temporary assumptions and relax them
in Section~\ref{sec:appr1} and Section~\ref{sec:appr2}, respectively.
For now, assume that all weights are equal and consider only the
threshold variant $T = 0$.
The task is to find all vertex pairs that form an edge, i.e., their distance is below the threshold $c(w_uw_v/W)^{1/d}$.
Since all weights are equal, the threshold in this restricted scenario is the same for all vertex pairs.
One approach to quickly identify adjacent vertices is to partition the
ground space into a grid of cells.  The size of the cells should be
chosen, such that
\begin{inparaenum}[(I)]
\item \label{item:small-cells}the cells are as small as possible and
\item \label{item:cells-of-suff-diam}the diameter of cells is larger than the threshold $c(w_uw_v/W)^{1/d}$.
\end{inparaenum}
The latter implies that only vertices in neighboring cells can be
connected thus narrowing down the search space.  The former ensures
that neighboring cells contain as few vertex pairs as possible
reducing the number of comparisons. Figure~\ref{fig:cellsa} shows an example of
such a grid for a $2$-dimensional ground space.

\subsection{Inhomogeneous Weights}
\label{sec:appr1}

\begin{figure}
  \centering
  \begin{subfigure}[t]{0.19\textwidth}
    \centering
    \includegraphics{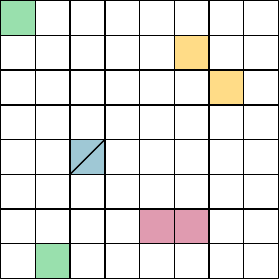}
    \caption{\centering}
    \label{fig:cellsa}
  \end{subfigure}
  \quad
  \begin{subfigure}[t]{0.19\textwidth}
    \centering
    \includegraphics{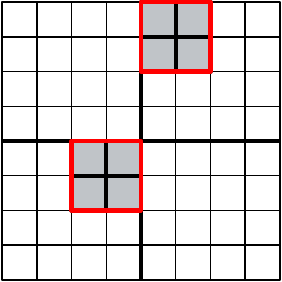}
    \caption{\centering}
    \label{fig:cellsb}
  \end{subfigure}
  \quad
  \hspace{3pt}
  \begin{subfigure}[t]{0.45\textwidth}
    \centering
    \includegraphics{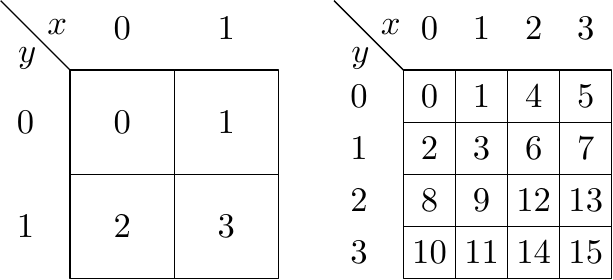}
    \caption{\centering}
    \label{fig:cell-linearization}
  \end{subfigure}
  \vspace{-2pt}
  \caption{\textbf{(\subref{fig:cellsa}),(\subref{fig:cellsb})}~The grid used by weight
    bucket pairs with a connection probability threshold between
    $2^{-3}$ and $2^{-4}$ in two dimensions. \textbf{(\subref{fig:cellsa})} Each pair of colored
    cells represent neighbors.  Note that the ground space is a torus
    and a cell is also a neighbor to itself. \textbf{(\subref{fig:cellsb})} The eight gray cells
    represent multiple distant cell pairs, which are replaced by one
    pair consisting of the red outlined parent cell pair.
    \textbf{(\subref{fig:cell-linearization})}~Linearization of the
    cells on level 1 (left) and 2 (right) for $d = 2$.}
\end{figure}

Assume that we have vertices with two different weights $w_1, w_2$,
rather than one.  As before, the cells should still be as small as
possible while having a diameter larger than the connection threshold.
However, there are three different thresholds now, one for each
combination of weights.  To resolve this, we can group the vertices by
weight and use three differently sized grids to find the edges between
them.

As GIRGs require not only two but many weights, considering one grid
for every weight pair is infeasible.  The solution is to discretize
the weights by grouping ranges of weights into \emph{weight buckets}.
When searching for edges between vertices in two weight buckets, the
pair of largest weights in these buckets provides the threshold for
the cell diameter.  This choice of the cell diameter satisfies
property~(\ref{item:cells-of-suff-diam}).
Property~(\ref{item:small-cells}) is violated only slightly, if the
weight range within the bucket is not too large.  Thus, each
combination of two weight buckets uses a grid of cells, whose
granularity is based on the maximum weight in the respective buckets.

There is a tradeoff when choosing the number of weight buckets.
Logarithmically many buckets yield a sublinear number of grids.
Moreover, the largest and smallest weight in a bucket are at most a
factor~$2$ apart.  Thus, the diameter of a cell is too large by a
factor~$\le 4$.

With this approach, a single vertex has to appear in grids of
different granularity.  To do this in an efficient manner, we
recursively divide the space into ever smaller grid cells, leading to
a hierarchical subdivision of the space.  This hierarchy is naturally
described by a tree.  For a $2$-dimensional ground space, each node has
four children, which is why we call it \emph{quadtree}.  Note that
each level of the quadtree represents a grid of different granularity.
Moreover, the side length of a grid cell on level $\ell$ is
$2^{-\ell}$.  For a pair $(i, j)$ of weight buckets, we then choose
the level that fits best for the corresponding weights, i.e., the
deepest level such that the diameter of each grid cell is above the
connection threshold for the largest weights in bucket $i$ and $j$,
respectively.  We call this level the \emph{comparison level}, denoted
by $\comparisonLevel{i}{j}$.  It suffices to insert vertices of a
bucket into the deepest level among all its comparison levels.  This
level is called the \emph{insertion level} and we denote it by
$\insertionLevel{i}$.  In Section~\ref{sec:cellaccess}, we discuss in
detail how to efficiently access all vertices in a given grid cell
belonging to a given weight bucket.

\subsection{Binomial Variant of the Model}
\label{sec:appr2}

For $T > 0$, neighboring cell pairs are still easy to handle: a
constant fraction of vertex pairs will have an edge and one can sample
them by explicitly checking every pair.  For distant cell pairs and a
fixed pair of weight buckets, the distance between the cells yields an
upper bound on the connection probability of included vertices; see
Equation~(\ref{eq:girg-binom}).  The probability bound depends on
both, the weight buckets and the cell pair distance, using the maximum
weight within the buckets and the minimum distance between points in
the cells.  We note that, the individual connection probabilities are
only a constant factor smaller than the upper bound.

Knowing this, we can use geometric jumps to skip most vertex pairs
\cite{ad-srs-85}.  The approach works as follows.  Assume that we want
to create an edge with probability $\overline{p}$ for each vertex
pair.  For this process, we define the random variable $X$ to be the
number of vertex pairs we see until we add the next edge.  Then $X$
follows a geometric distribution.  Thus, instead of throwing a coin
for each vertex pair, we can do a single experiment that samples $X$
from the geometric distribution and then skip $X$ vertex pairs ahead.
Since not all vertex pairs reach the upper bound $\overline{p}$, we
accept encountered pairs with probability $p_{uv}/\overline{p}$ to get
correct results.

Although distant cell pairs are handled efficiently, their number is
still quadratic, most of which yield no edges.  To circumvent this
problem, the sampling algorithm, yet again, uses a quadtree.  In the
quadratic set of cell pairs to compare for one weight bucket pair,
non-neighboring cells are grouped together along the quadtree
hierarchy.  They are replaced by their parents as shown in
Figure~\ref{fig:cellsb}.  The grouping of distant cell pairs is done as
much as possible, meaning as long as the parents are not neighbors.

In conclusion, for each pair of weight buckets $(i,j)$ the following two
types of cell pairs have to be processed.  Any two neighboring cell
pairs on the comparison level $\comparisonLevel{i}{j}$; and any
distant cell pair with level larger or equal $\comparisonLevel{i}{j}$
that has neighboring parents.

\subsection{Efficiently Iterating Over Cell Pairs}
\label{sec:loopswap}

The previous description sketches the algorithm as originally
published.  Here, we propose a refactoring that greatly simplifies
the implementation and enables parallelization.  We attribute a
significant amount of \hypergirgs' speed up over \embedder to this
change.

Instead of first iterating over all bucket pairs and then over all
corresponding cell pairs, we reverse this order.  This removes the
need to repeatedly determine the cell pairs to process for a given
bucket pair.  Instead it suffices to find the bucket pairs that
process a given cell pair.  This only depends on the level of the two
cells and their type (neighboring or distant).  Inverting the mapping
from bucket pairs to cell pairs in the previous section yields the
following.  A neighboring cell pair on level $\ell$ is processed for
bucket pairs with a comparison level of exactly $\ell$.  A distant
cell pair on level $\ell$ (with neighboring parents) is processed for
bucket pairs with a comparison level larger than or equal to $\ell$.
Thus, for each level of the quadtree we must enumerate all neighboring
cell pairs, as well as distant cell pairs with neighboring parents.
Algorithm~\ref{alg:recursive} recursively enumerates exactly these
cell pairs.

\begin{algorithm2e}
\DontPrintSemicolon
\caption{Sample GIRG by Recursive Iteration of Cell Pairs}
\label{alg:recursive}
\KwIn{cell pair (A,B)}
\SetKwFunction{recur}{recur}
\ForAll{bucket pairs $(i,j)$ that process the cell pair $(A,B)$}{
    \eIf{A and B are neighbors}{
        emit each edge $(u,v) \in V_i^A \times V_j^B$ with probability $p_{uv}$\;
    }{
        choose candidates $S\subseteq V_i^A \times V_j^B$ using geometric jumps and $\overline{p}$\;
        emit each edge $(u,v)\in S$ with probability $p_{uv}/\overline{p}$\;
    }
}
\If{A and B are neighbors \emph{\textbf{and not}} maximum depth reached}{
	\ForAll{children $X$ of $A$}{
		\ForAll{children $Y$ of $B$}{
			\recur{X,Y}
		}
	}
}
\end{algorithm2e}

\subsection{Efficient Access to Vertices by Bucket and Cell}
\label{sec:cellaccess}

A crucial part of the algorithm is to quickly access the set of
vertices restricted to a weight bucket $i$ and a cell $A$, which we
denote by $V_i^A$.  To this end, we linearize the cells of each level
as illustrated in Figure~\ref{fig:cell-linearization}.  This
linearization is called Morton code~\cite{m-cogdb-66} or z-order
curve~\cite{om-cdsas-84}.  It has the nice properties that
\begin{inparaenum}[(I)]
\item for each cell in level $\ell$, its descendants in level
  $\ell' > \ell$ in the quadtree appear consecutively; and
\item it is easy to convert between a cells position in the linear
  order and its $d$-dimensional coordinates (see
  Section~\ref{sec:morton}).
\end{inparaenum}

We sort the vertices of a fixed weight bucket $i$ by the Morton code
of their containing cell on the insertion level $\insertionLevel{i}$,
using arbitrary tie-breaking for vertices in the same cell.  This has
the effect that for any cell $A$ with
$\level(A) \le \insertionLevel{i}$, the vertices of $V_i^A$ appear
consecutive.  Thus, to efficiently enumerate them, it suffices to know
for each cell $A$ the index of the first vertex in $V_i^A$.  This can
be precomputed using prefix sums leading to the following lemma.


\wormhole{lemma-access}
\begin{lemma}
\label{lemma:access}
After linear preprocessing, for all cells $A$ and weight buckets $i$
with $\level(A) \leq \insertionLevel{i}$, vertices in the set $V_i^A$
can be enumerated in $\mathcal{O}(|V_i^A|)$.
\end{lemma}

\section{Implementation Details}
\label{sec:impl}

The description in the previous section is an idealized version of the
algorithm.  For an actual implementation, there are some gaps to fill
in.  Moreover, omitting many minor tweaks, we want to sketch
optimizations that are crucial to achieve a good practical run time in
the following.  More details on the sketched approaches can be found
in Appendix~\ref{sec:impl-deta}.

\subsection{Estimating the Average Degree Parameter}
\label{sec:constants}

Here, we sketch how to estimate the parameter $c$ in
Eq.~(\ref{eq:girg-binom}) to achieve a given expected average degree.
We estimate the constant based on the actual weights, not on their
probability distribution.  This leads to lower variance and allows
user-defined weights.

We start with an arbitrary constant $c$, calculate the resulting
expected average degree $\mathbb E[\bar{d}]$ and adjust $c$
accordingly, using a modified binary search.  This is possible, as
$\mathbb E[\bar{d}]$ is monotone in $c$.  We derive an exact formula
for $\mathbb E[\bar{d}]$, depending on $c$ and the weights.  It cannot
simply be solved for $c$, which is why we use binary search instead of
a closed expression.

For the binary search, we need to evaluate $\mathbb E[\bar{d}]$ for
different values of $c$.  This is potentially problematic, as the
formula for $\mathbb E[\bar{d}]$ sums over all vertex pairs.  The
issue preventing us from simplifying this formula is the minimum in
the connection probability.  We solve this, by first ignoring the
minimum and subtracting an error term for those vertex pairs, where
the minimum takes effect.  The remaining hard part is to calculate
this error term.  Let $E_R$ be the set of vertex pairs appearing in
the error term and let $R$ be the set of vertices with at least one
partner in $E_R$.  Although $|E_R|$ itself is sufficiently small, $R$
is too large to determine $E_R$ by iterating over all pairs in
$R \times R$.  We solve this by iterating over the vertices in $R$,
sorted by weight.  Then, for each vertex we encounter, the set of
partners in $E_R$ is a superset of the partners of the previous vertex
(with smaller weight).

\subsection{Efficiently Encoding and Decoding Morton Codes}
\label{sec:morton}

Recall from Section~\ref{sec:cellaccess} that we linearize the
$d$-dimensional grid of cells using Morton code.  As vertex positions
are given as $d$-dimensional coordinates, we have to convert the
coordinates to Morton codes (i.e., the index in the linearization) and
vice versa.  This is done by bitwise interleaving the coordinates.
For example, the $2$-dimensional Morton code of the four-bit coordinates
$a=a_3a_2a_1a_0$ and $b=b_3b_2b_1b_0$ is $a_3b_3a_2b_2a_1b_1a_0b_0$.
We evaluated different encoding and decoding approaches via micro
benchmarks.  The fastest approach, at least on Intel processors, was
an assembler instruction from BMI2 proposed by Intel in
2013~\cite{intel-manual-19}.

\subsection{Generating HRGs Avoiding Expensive Mathematical Operations}
\label{sec:nomath}

The algorithm from Section~\ref{sec:algorithm} can be used to generate HRGs. The
algorithm works conceptually the same, except that most formulas
change.  This has for example the effect that we no longer get a
closed formula to determine the insertion level of a weight bucket or
the comparison level of a bucket pair.  Instead, one has to search
them, by iterating over the levels of the quadtree.  Furthermore, HRGs 
introduce many computationally expensive mathematical operations 
like the hyperbolic cosine.  This can be mitigated as follows.

For the threshold model, an edge exists if the distance $d$ is smaller
than $R$.  Considering how the hyperbolic distance is defined
(Section~\ref{sec:hrgmodel}), reformulating it to
$\cosh(d) < \cosh(R)$ avoids the expensive $\acosh$, while $\cosh(R)$
remains constant during execution and can thus be precomputed.
Similar to recent threshold HRG generators, we compute intermediate
values per vertex such that $\cosh(d)$ can be computed using only
multiplication and addition~\cite{flm-cfmdgg-17,p-gprhg-17}.

For the binomial model, evaluating the connection probability is a
performance bottleneck.  The straightforward way to sample edges is:
compute the connection probability $p_T(d)$ depending on the distance,
sample a uniform random value $u \in [0,1]$, and create the edge if
and only if $u < p_T(d)$.  We can improve this by precomputing the
inverse of $p_T(d)$ for equidistant values in $[0, 1]$.  This lets us,
for small ranges in $[0, 1]$, quickly access the corresponding range
of distances.  Changing the order, we first sample $u \in [0, 1]$, which falls in a range
between two precomputed values, which in turn yields a range of
distances.  If the actual distance lies below that range, there has to
be an edge and if it lies above, there is no edge.  Only if it lies in
the range, we actually have to compute the probability $p_T(d)$.

\subsection{Parallelization}
\label{sec:parallel}

The algorithm has five steps: generate weights, generate positions,
estimate the average degree constant, precompute the geometric data
structure, and sample edges.  The first two are trivial to
parallelize.  For estimating the constants, we parallelize the
dominant computations with linear running time.
To sample the edges, we make use of the fact that we iterate over
cell pairs in a recursive manner.  This can be parallelized by cutting
the recursion tree at a certain level and distributing the loose ends
among multiple processors.

For the preprocessing we have to do three subtasks: compute for each
vertex its containing cells on its insertion level, sort the vertices
according to their Morton code index, and compute the prefix sum for
all cells.  We parallelize all three tasks and optimize them by
handling all weight buckets together, sorting by weight bucket first
and Morton code second.  This is done by encoding this criterion into
integers that are sorted with parallel radix sort.



\section{Experimental Evaluation}
\label{sec:experiments}

We perform three types of experiments.  In
Section~\ref{sec:girgexperiments} we investigate the scaling behavior
of our GIRG generator, broken down into the different tasks performed
by the algorithm.  In Section~\ref{sec:hrgexperiments} we compare our
HRG generator with existing generators.  In
Section~\ref{sec:hrgisgirgexperiments} we experimentally investigate
the difference between HRGs and their GIRG counterpart.
Whenever a data point represents the mean over multiple iterations,
our plots include error-bars that indicate the standard deviation.
Besides the implementation itself, all benchmarks and analysis scripts
are also accessible in our source repository.

\subsection{Scaling of the GIRG Generator}
\label{sec:girgexperiments}

We investigate the scaling of the generator, broken down into five
steps.
\begin{inparaenum}[1.]
\item \textbf{(Weights)}~Generate power-law weights.
\item \textbf{(Positions)} Generate points on~$\mathbb{T}^d$.
\item \textbf{(Binary)} Estimate the constant controlling the average
  degree.
\item \textbf{(Pre)} Preprocess the geometric data structure
  (Section~\ref{sec:cellaccess}).
\item \textbf{(Edges)} Sample edges between all vertex pairs as described in
  Algorithm~\ref{alg:recursive}.
\end{inparaenum}

\begin{figure}
  \centering
  \begin{subfigure}[t]{0.44\textwidth}
    \centering
    \includegraphics[width=\textwidth]{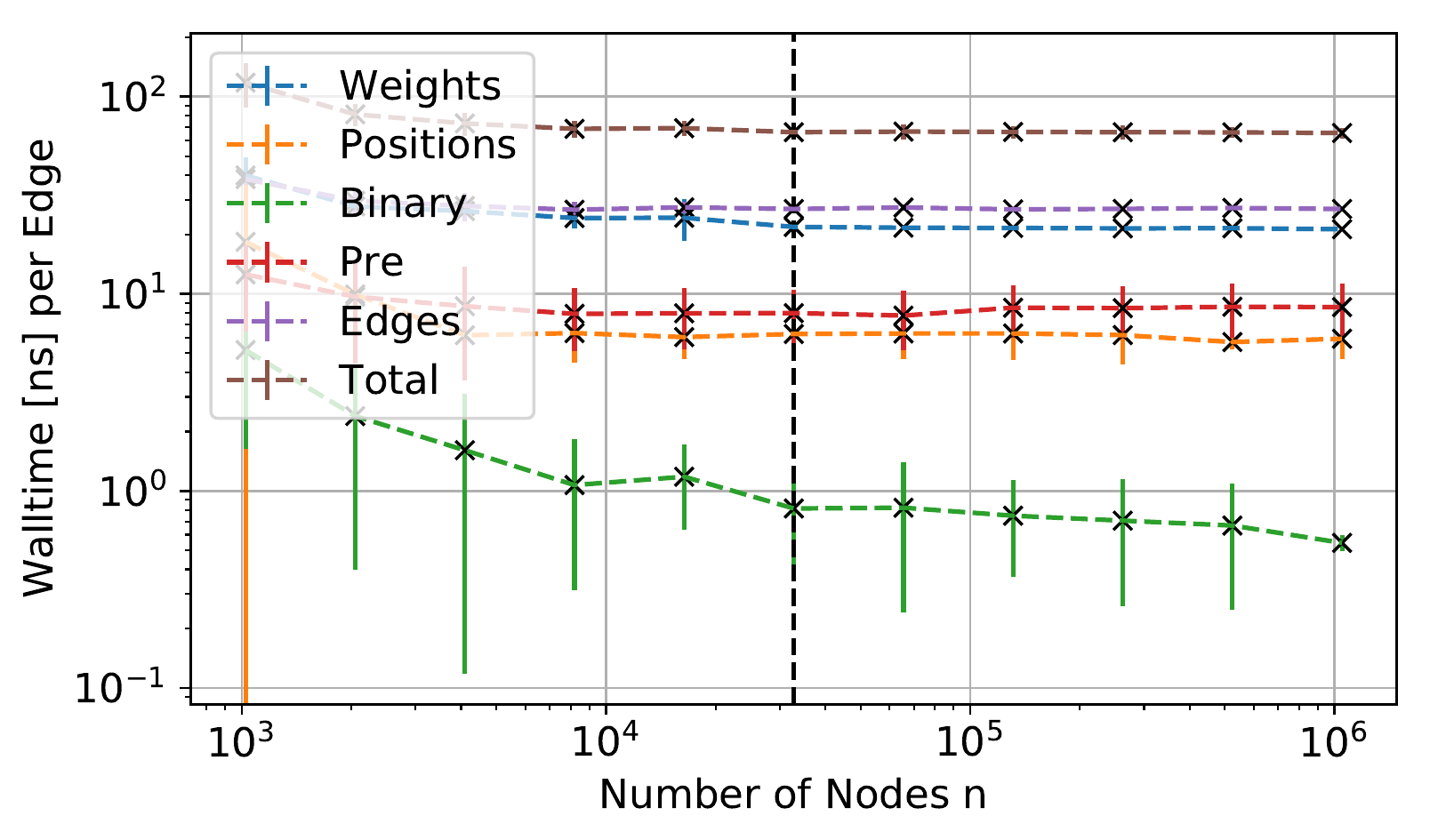}
  \end{subfigure}
  \quad
  \begin{subfigure}[t]{0.44\textwidth}
    \centering
    \includegraphics[width=\textwidth]{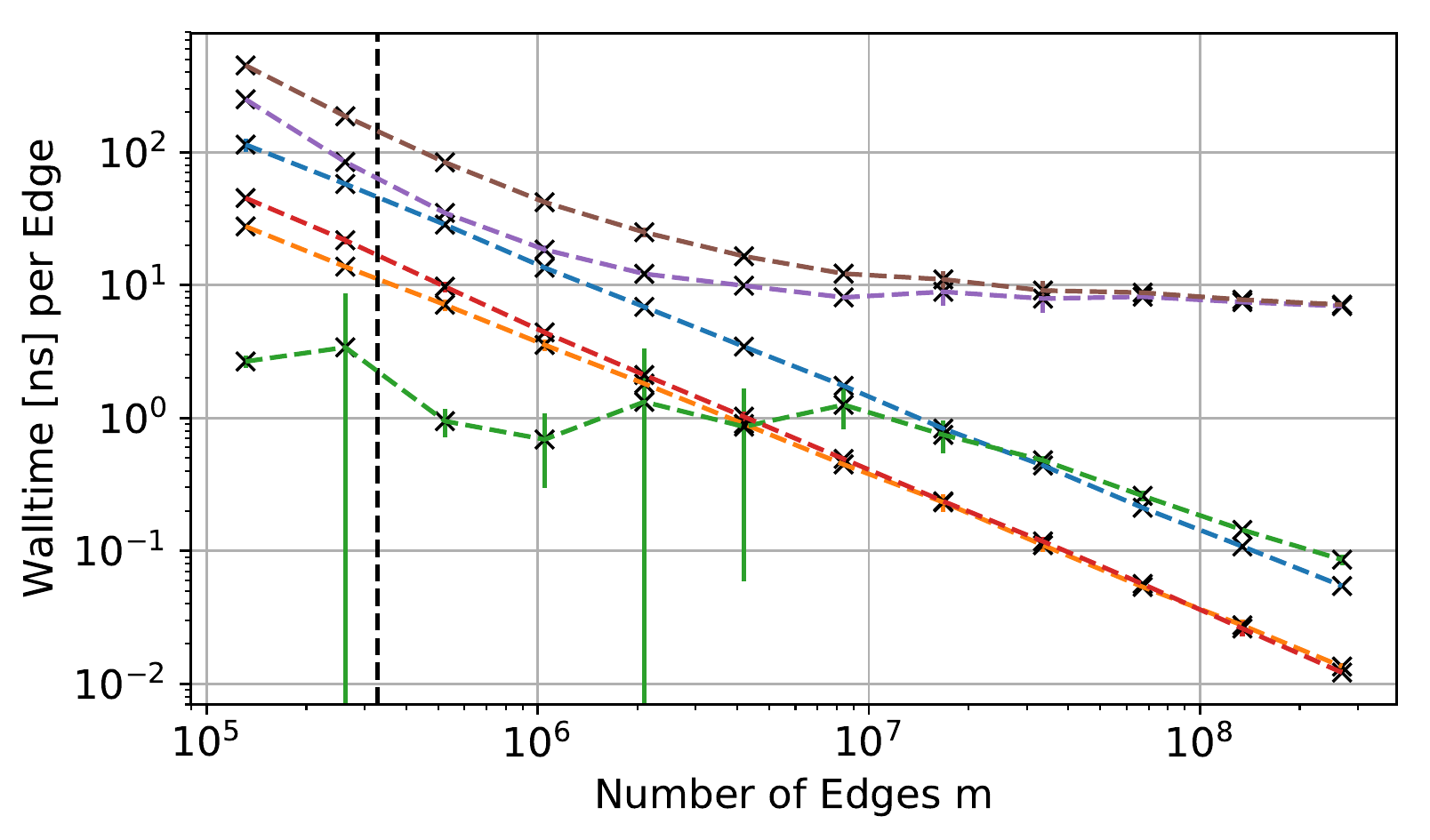}
  \end{subfigure}
  \begin{subfigure}[t]{0.44\textwidth}
    \centering
    \includegraphics[width=\textwidth]{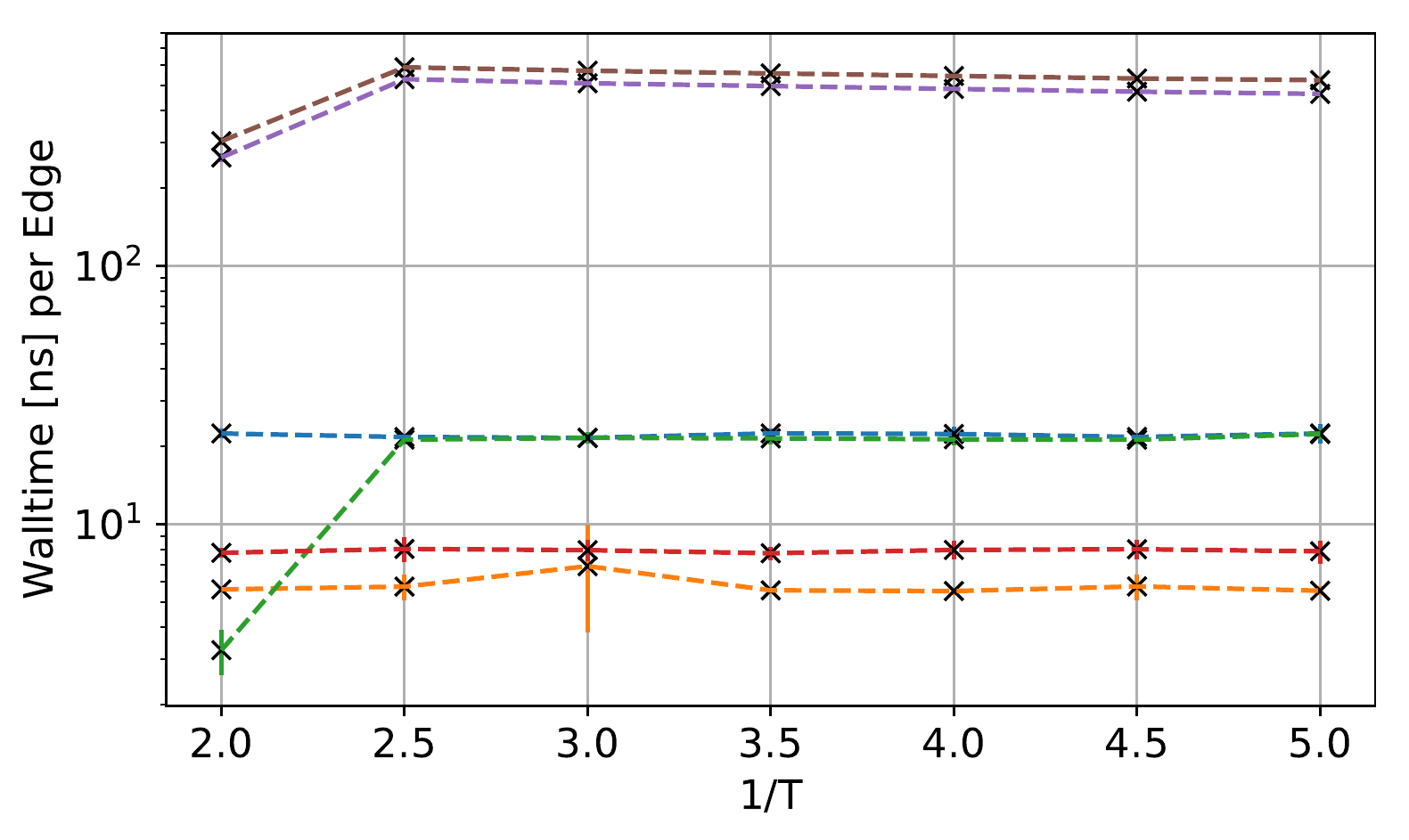}
  \end{subfigure}
  \quad
  \begin{subfigure}[t]{0.44\textwidth}
    \centering
    \includegraphics[width=\textwidth]{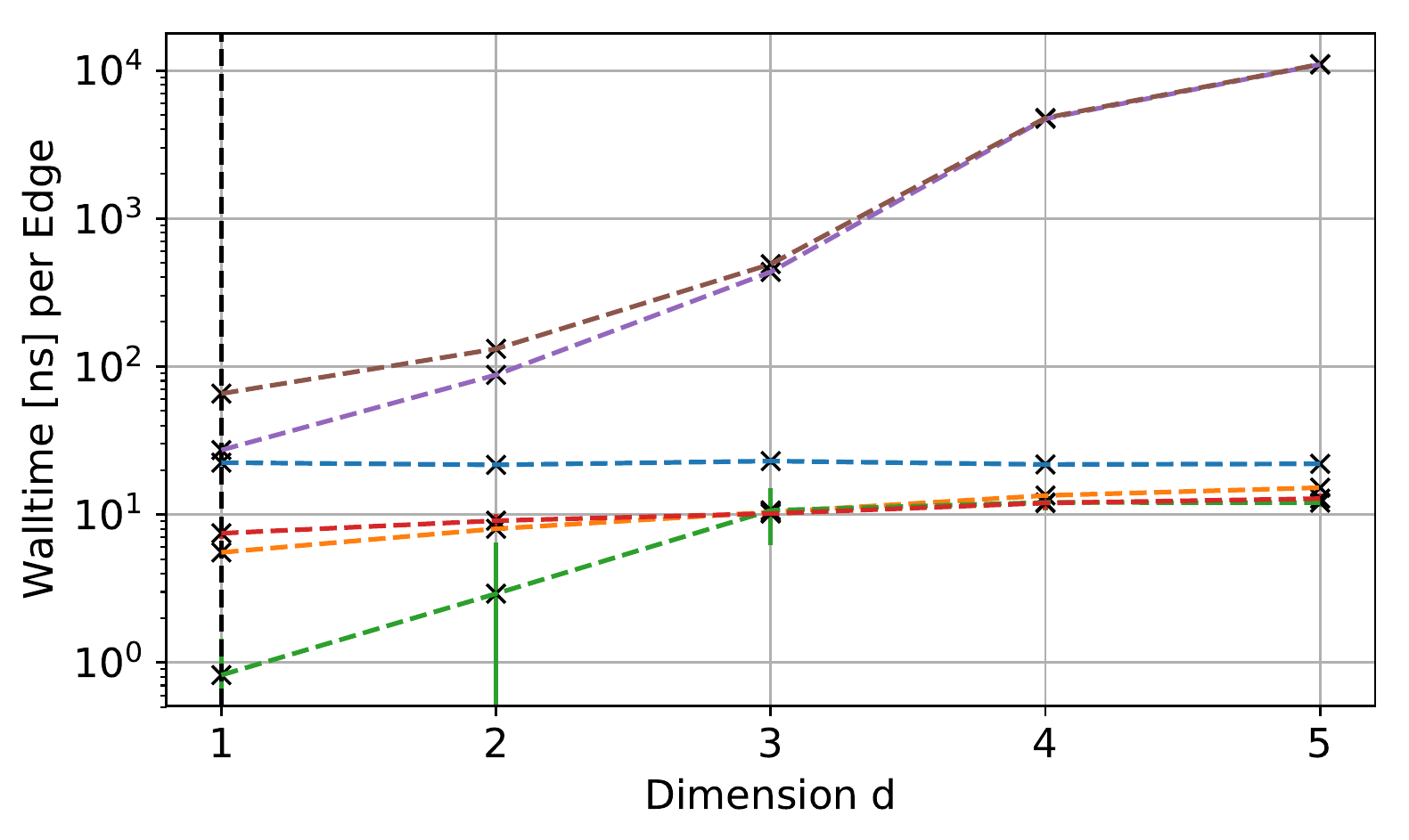}
  \end{subfigure}
  \caption{Run time for the steps of the GIRG sampling algorithm averaged over 10 iterations. 
    Each plot varies a different model parameter deviating from a fixed base configuration with $d=1$, $n=2^{15}$, $T=0$, $\beta=2.5$, and $\bar d=10$.
    The base configuration is indicated by a dashed vertical line.  
  }
  \label{fig:girgeval}
\end{figure}

Figure~\ref{fig:girgeval} shows the sequential run time over the
number of nodes $n$ (top left), number of edges $m$ (top right),
temperature $T$ (bottom right), and dimension $d$ (bottom right).  The
performance is measured in nanoseconds per edge.  Each data point
represents the mean over 10 iterations.  To make the measurements
independent of the graph representation, we do not save the edges into
RAM, but accumulate a checksum instead.  Note that the top right plot
increases the average degree, resulting in a decreased time per edge.

The empirical run times match the theoretical bounds: it is linear in
$n$ and $m$, grows exponentially in the dimension $d$, and is
unaffected by the temperature $T$.  The overall time is dominated by
the edge sampling.  Generating the weights includes expensive
exponential functions, making it the slowest step after edge sampling.
Generating the positions is significantly faster even for higher
dimensions.  For the parameter estimation using binary search, one can
see that the run time never exceeds the time to generate the weights.
For non-zero temperature~$T$ the performance of the binary search is
similar to the generation of the weights, as it also requires
exponential functions.  The lower run times per edge for the
increasing number of edges (top right) show that the run time is
dominated by the number of nodes $n$.  Only for very high average
degrees, the cost per edge outgrows the cost per vertex.

\subsection{HRG Run Time Comparison}
\label{sec:hrgexperiments}

We evaluate the run time performance of \hypergirgs compared to the
generators in Table~\ref{tab:related}, excluding the generators with
high asymptotic run time as well as \rhg and \srhg.  \rhg and \srhg
are designed for distributed machines.  Executed on a single compute
node, the performance of the faster \srhg is comparable to
\hypergen~\cite{flm-cfmdgg-17}.
To avoid systematic biases between different graph representations,
the implementations are modified\footnote{%
  The modifications are publicly available and referenced in our
  GitHub repository.}  not to store the resulting graph.  Instead,
only the number of edges produced is counted and we ensure that the
computation of incident nodes is not optimized away by the compiler.

We used different machines for our sequential and parallel
experiments.  The former are done on an \textsl{Intel Core i7-8700K}
with \SI{16}{GB} RAM, the latter on an \textsl{Intel Xeon CPU
  E5-2630~v3} with 8 cores (16 threads) and \SI{64}{GB} RAM.

\begin{figure}[tb]
  \begin{subfigure}[t]{0.44\textwidth}%
    \centering%
    \includegraphics[width=\textwidth]{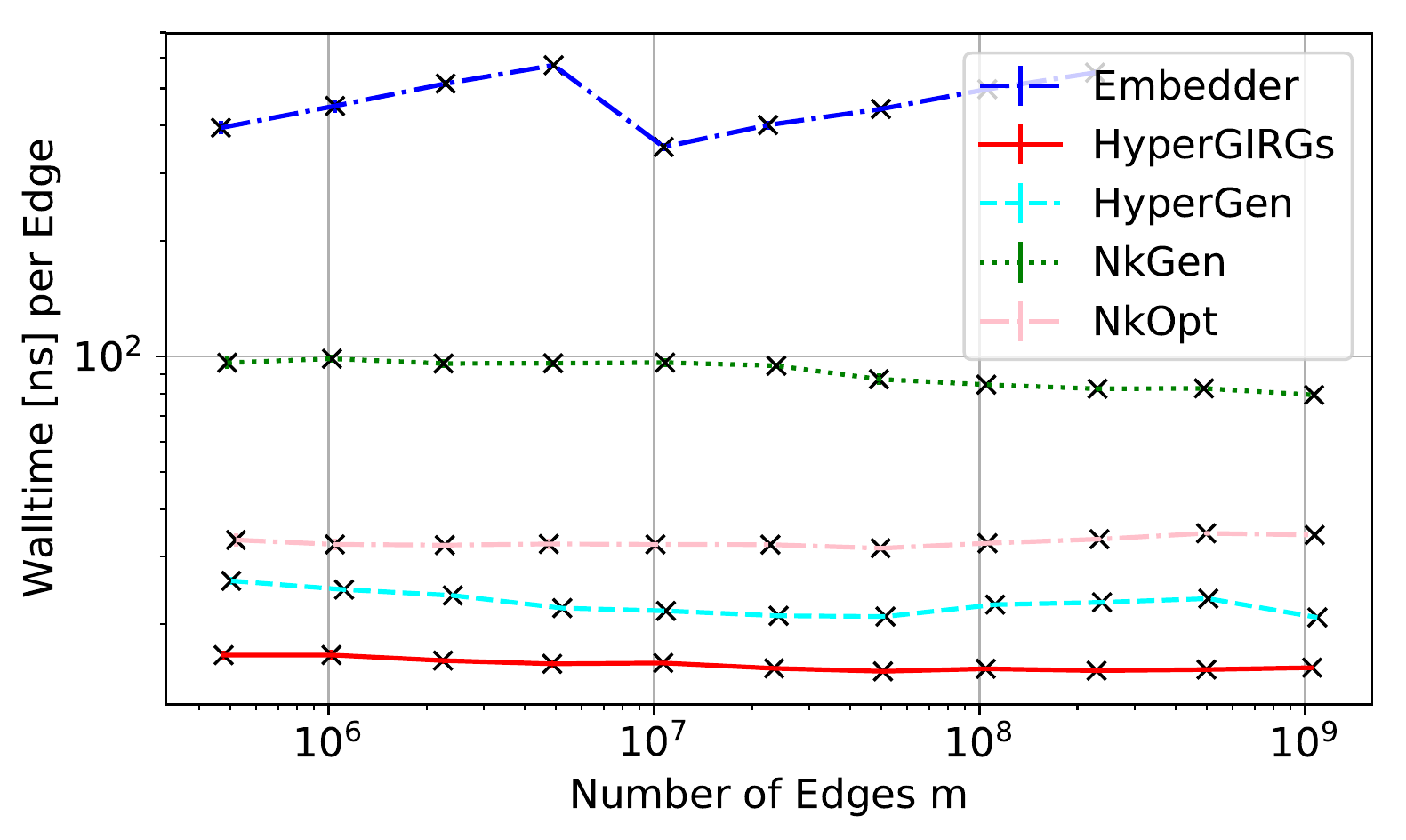}%
    \caption{\centering$\bar d = 100$, $\beta = 2.2$, $T = 0$,
      sequential}%
    \label{subfig:runtime-rhg-22-0}
  \end{subfigure}\hfill
  \begin{subfigure}[t]{0.44\textwidth}%
    \includegraphics[width=\textwidth]{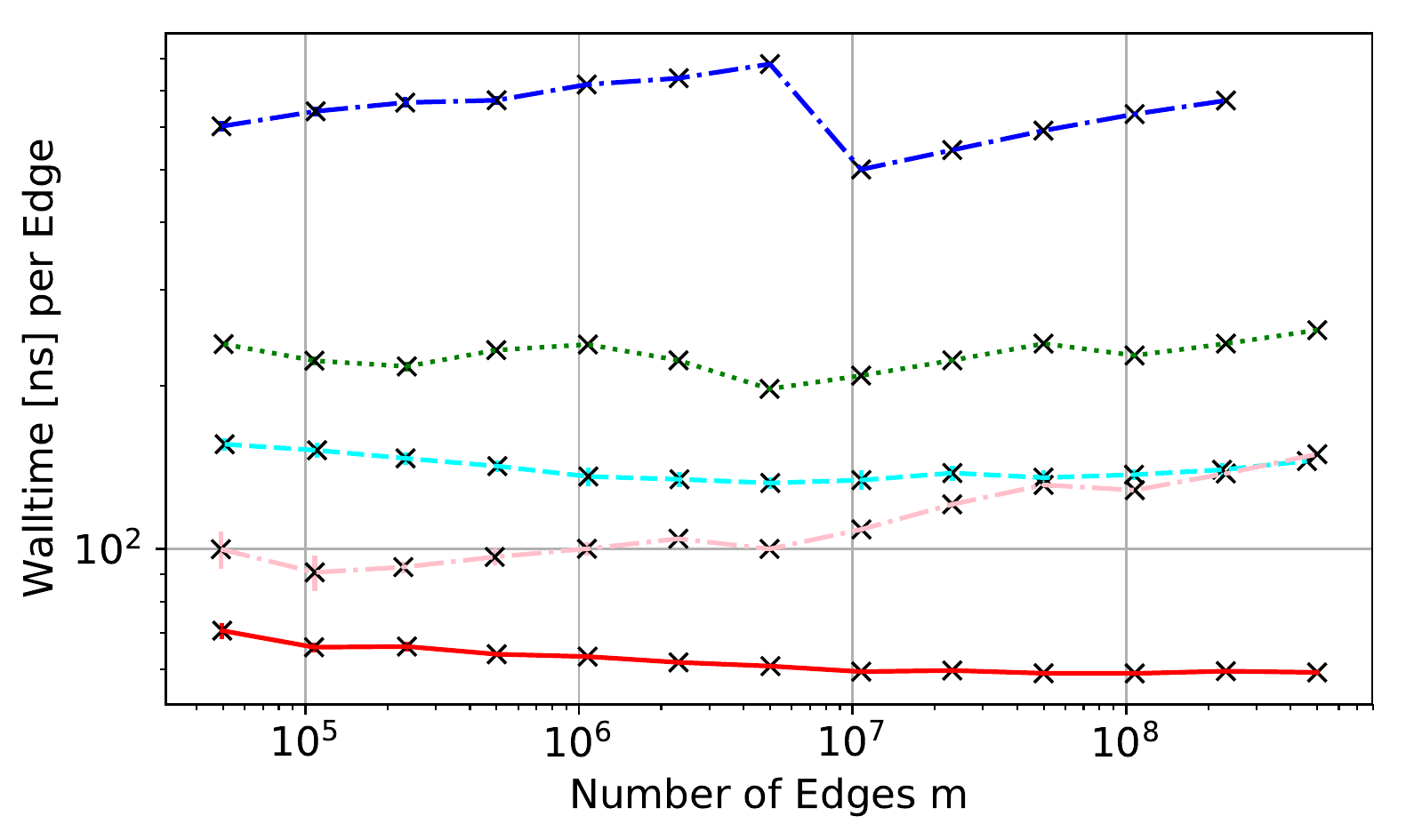}%
    \caption{\centering$\bar d = 10$, $\beta = 3$, $T = 0$,
      sequential}%
    \label{subfig:runtime-rhg-30-0}
  \end{subfigure}
  
  \begin{subfigure}[t]{0.44\textwidth}%
    \includegraphics[width=\textwidth]{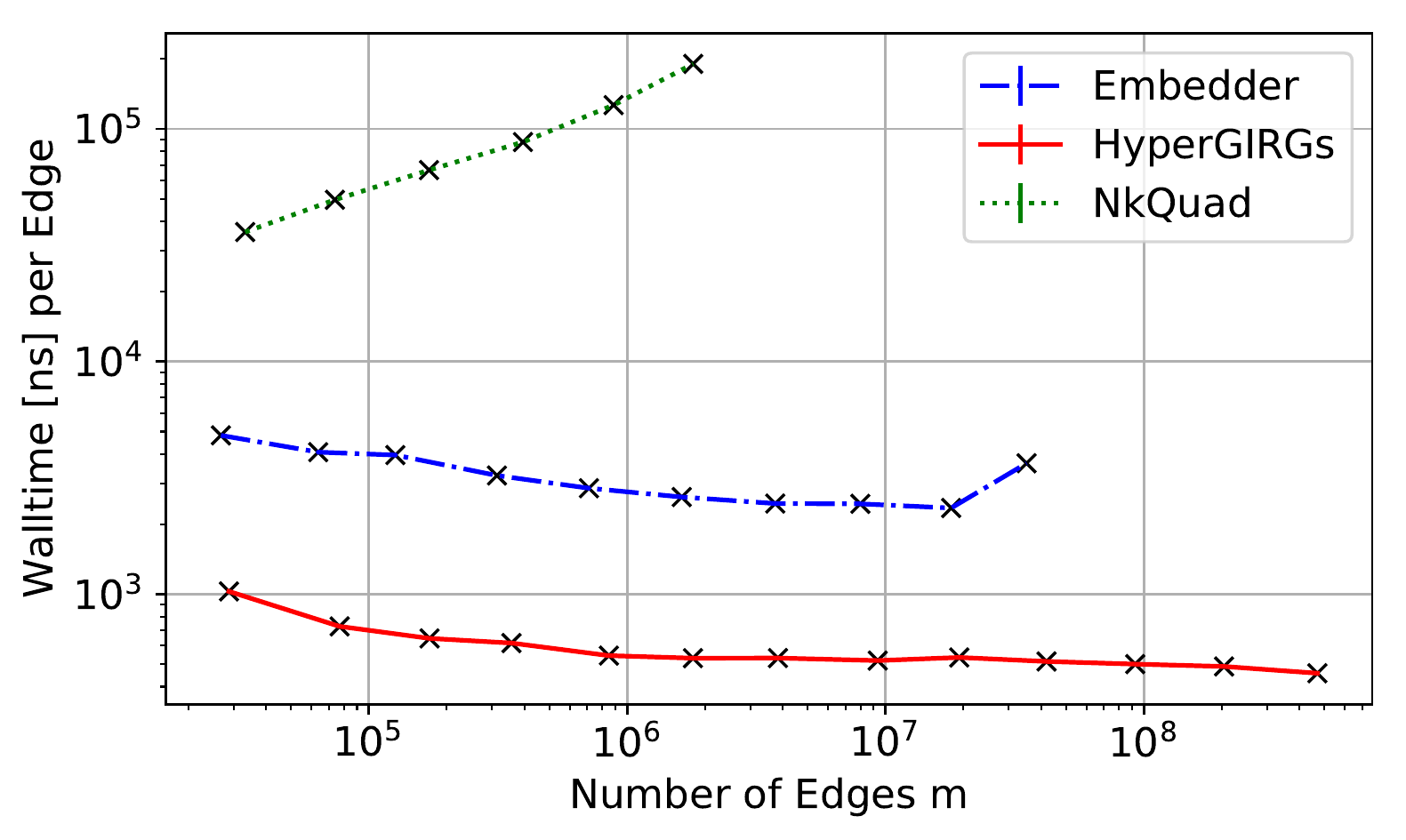}%
    \caption{\centering$\bar d = 10$, $\beta = 2.2$, $T = 0.5$,
      sequential}%
    \label{subfig:runtime-rhg-22-5}
  \end{subfigure}  \hfill
  \begin{subfigure}[t]{0.44\textwidth}%
    \includegraphics[width=\textwidth]{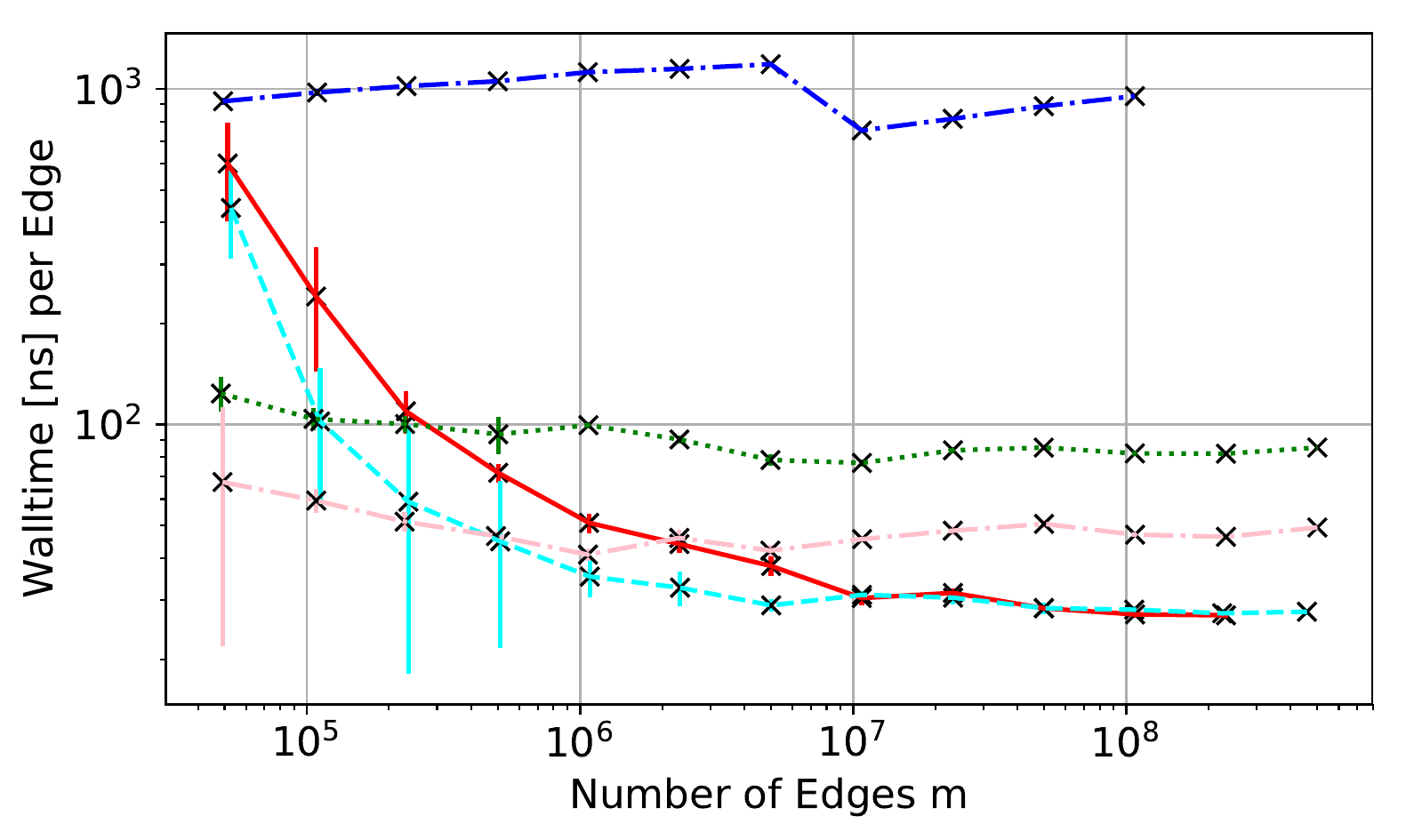}%
    \caption{\centering$\bar d = 10$, $\beta = 3$, $T = 0$, parallel
      (16 threads)}%
    \label{subfig:runtime-rhg-30-0-par}
  \end{subfigure}
  \caption{Comparison of HRG generators averaged over 5 iterations.
    \textbf{(\subref{subfig:runtime-rhg-22-0})},
    \textbf{(\subref{subfig:runtime-rhg-30-0})} Threshold variant for
    different average degrees~$\bar{d}$ and power-law
    exponents~$\beta$.
    \textbf{(\subref{subfig:runtime-rhg-22-5})}~Binomial variant with
    temperature $T = 0.5$.
    \textbf{(\subref{subfig:runtime-rhg-30-0-par})}~The same
    configuration as (\subref{subfig:runtime-rhg-30-0}) but utilizing
    multiple cores.}
  \label{fig:runtime-rhg}
\end{figure}


Our generator \hypergirgs is consistently faster than the competitors,
independent of the parameter choices; see
Figure~\ref{subfig:runtime-rhg-22-0}
and~\ref{subfig:runtime-rhg-30-0}.  Only for unrealistic average
degrees (\SI{1}{k}), \hypergen slightly outperforms \hypergirgs.
Moreover, \hypergirgs beats \embedder, the only other efficient generator
supporting non-zero temperature, by an order of magnitude.


For higher temperatures, we compare our algorithm with the two other
non-quadratic generators \nkquad (included in NetworKit) and
\embedder; see Figure~\ref{subfig:runtime-rhg-22-5}.  We note that
\embedder uses a different estimation for $R$, which leads to an
insignificant left-shift of the corresponding curve.
In Figure~\ref{subfig:runtime-rhg-22-5}, one can clearly see the worse
asymptotic running time of \nkquad.  Compared to \embedder, \hypergirgs
is consistently $4$ times faster.


Figure~\ref{subfig:runtime-rhg-30-0-par} shows measurements for
parallel experiments using 16 threads.  The parameters coincide with
Figure~\ref{subfig:runtime-rhg-30-0}.  \embedder does not support
parallelization and is outperformed even more by the other generators.
The fastest generator in this multi-core setting is \hypergen, which
is specifically tailored towards parallel execution.  Nonetheless,
\hypergirgs shows comparable performance and outperforms the other two
generators \nkgen and \nkopt.  We note that even on parallel machines,
the sequential performance is of high importance: One often needs a
large collection of graphs rather than a single huge instance.  In
this case, it is more efficient to run multiple instances of a
sequential generator in parallel.

\subsection{Difference Between HRGs and GIRGs}
\label{sec:hrgisgirgexperiments}

Recall from Section~\ref{sec:comp-girgs-hrgs} that a HRG with average
degree $\dhrg$ has a corresponding GIRG sub- and supergraphs with
average degrees $\dgirg$ and $\Dgirg$, respectively.

\begin{figure}[tb]
  \centering
  \begin{subfigure}[t]{0.44\textwidth}
    \centering
    \includegraphics[width=\textwidth]{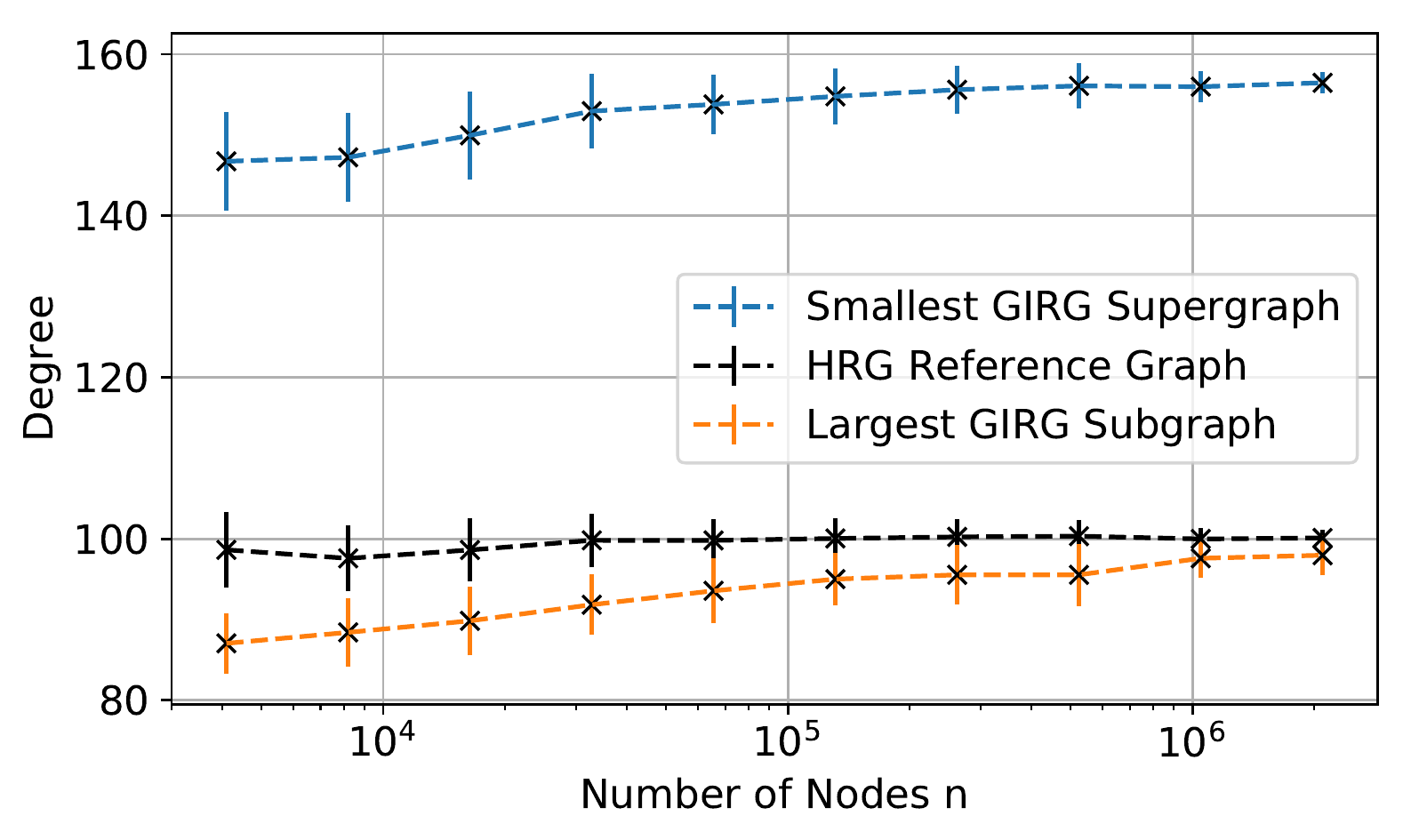}
    \caption{\centering$n \in [2^{12}, 2^{21}]$, $\bar d = 100$,
      $\beta = 2.5$, $T = 0$}
    \label{fig:hrgisgirg-range}
  \end{subfigure}
  \quad
  \begin{subfigure}[t]{0.44\textwidth}
    \centering
    \includegraphics[width=\textwidth]{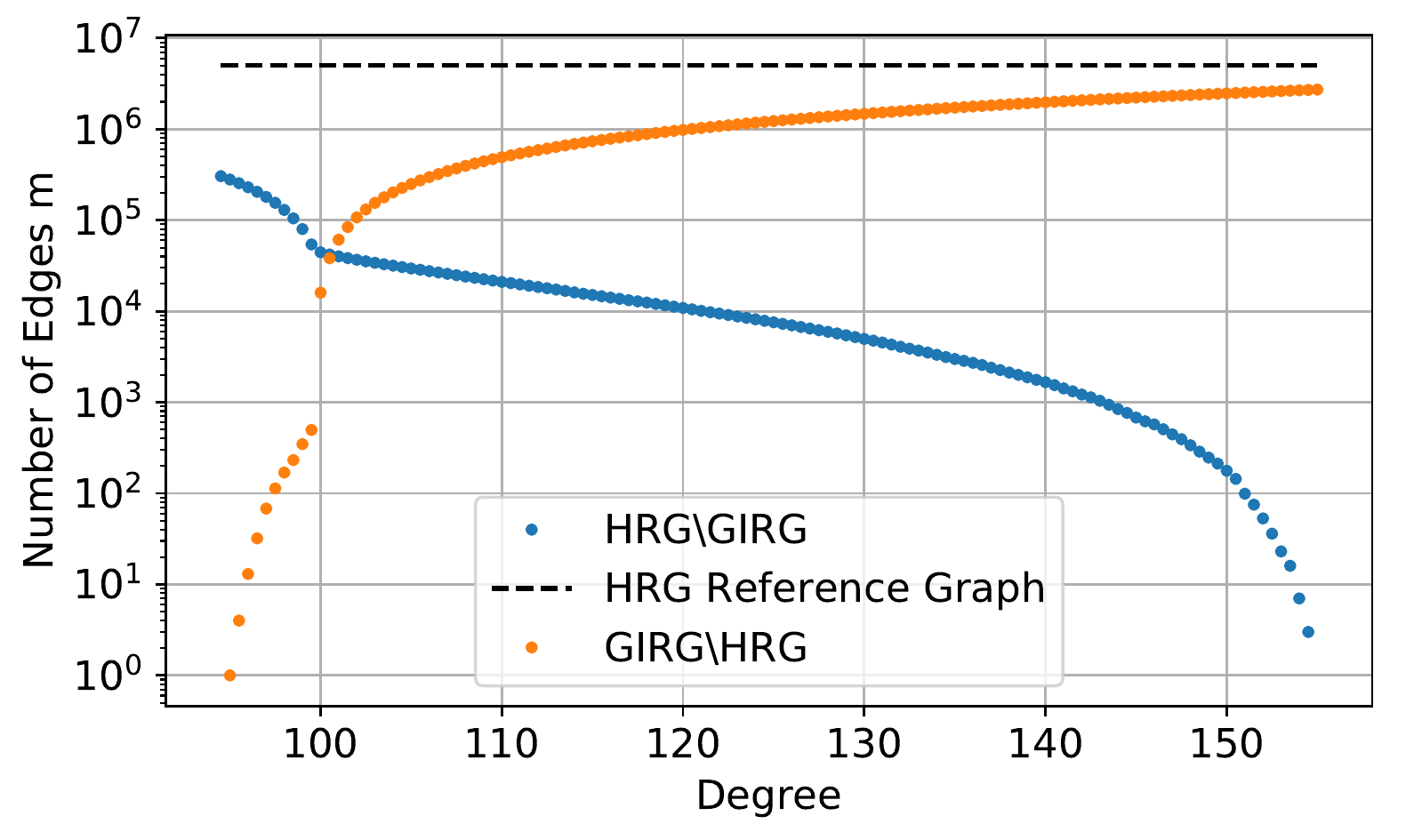}
    \caption{\centering$n = 10^6$, $\bar d = 100$, $\beta = 2.5$,
      $T = 0$}
    \label{fig:hrgisgirg-single}
  \end{subfigure}
  \caption{Relation between the HRG and the GIRG model.
    \textbf{(\subref{fig:hrgisgirg-range})}~The values for $\dhrg$,
    $\dgirg$, $\Dgirg$ averaged over 50 iterations.
    \textbf{(\subref{fig:hrgisgirg-single})}~The number of missing
    ($\text{HRG}\setminus\text{GIRG}$) and additional
    ($\text{GIRG} \setminus \text{HRG}$) edges depending on the
    expected degree of the corresponding GIRG. It can be interpreted
    as a cross-section of one iteration
    in~(\subref{fig:hrgisgirg-range}). }
  \label{fig:hrgisgirg}
\end{figure}

We experimentally determine, for given HRGs, the values for $\dgirg$
by decreasing the average degree of the corresponding GIRGs until it
is a subgraph of the HRG.  Analogously, we determine the value for
$\Dgirg$.  We focus on the threshold variant of the models, as this
makes the coupling between HRGs and GIRGs much simpler (the graph is
uniquely determined by the coordinates).
Figure~\ref{fig:hrgisgirg-range} shows $\dgirg$ and $\Dgirg$, compared
to $\dhrg$ for growing $n$.  One can see that $\dgirg$ and $\Dgirg$
are actually quite far apart.  They in particular do not converge to
the same value for growing $n$.  However, at least $\dgirg$ seems to
approach $\dhrg$.  This indicates that every HRG corresponds to a GIRG
subgraph that is missing only a sublinear fraction of edges.  On the
other hand, the average degree of the GIRG has to be increased by a
lot to actually contain all edges also contained in the HRG.

Figure~\ref{fig:hrgisgirg-single} gives a more detailed view for a
single HRG.  Depending on the average degree of the GIRG, it shows how
many edges the GIRG lacks and how many edges the GIRG has in addition
to the HRG.  For degree~$100$, the GIRG contains about \SI{38}{k}
additional and lacks about \SI{42}{k} edges.  These are rather small
numbers compared to the \SI{50}{M} edges of the graphs.


\newpage
\bibliography{paper}

\newpage
\appendix

\section{Omitted Proof of Lemma~\ref{lemma:access}}
\label{sec:omitted-proof-lemma-1}

\begin{backInTime}{lemma-access}
\begin{lemma}
After linear preprocessing, for all cells $A$ and weight buckets $i$
with $\level(A) \leq \insertionLevel{i}$, vertices in the set $V_i^A$
can be enumerated in $\mathcal{O}(|V_i^A|)$.
\end{lemma}
\begin{proof}
  As mentioned above, we have to sort the vertices $V_i$ of each
  weight bucket $i$ according to the index (Morton code) of the
  containing cell.  Clearly, the $d$-dimensional coordinates of the
  cell containing a given vertex is obtained in constant time by
  rounding.  From this one can obtain the index in constant time (also
  see Section~\ref{sec:morton}).  This can be done using, e.g., bucket
  sort with respect to this index to sort the vertices.  In the
  following, we refer to this sorted array with $V_i$.

  Besides these sorted arrays $V_i$ of vertices, one for each weight
  bucket $i$, we store for each cell $C$ at level $\insertionLevel{i}$
  the number of vertices preceding the vertices in cell $C$.  Note
  that this is simply the prefix sum of the number of vertices in all
  cells that come before cell $C$.  Denote this prefix sum of cell $C$
  with $P_C$.


  Now let $i$ be a weight bucket and let $A$ be a cell identifying the
  requested set of vertices $V_i^A$ (with
  $\level(A) \leq \insertionLevel{i}$).  Let $C_1, \dots, C_j$ be the
  descendants of cell $A$ at level $\insertionLevel{i}$, appearing in
  this order according to the Morton code.  Recall that the vertices
  in $C_1, \dots, C_j$ appear consecutive in the sorted array $V_i$.
  Thus, $V_i^A$ is given by the range $[P_{C_1},\dots, P_{C_{j+1}})$
  in $V_i$.



  In terms of running time, each weight bucket requires
  $\mathcal{O}(|V_i| + 2^{d\cdot\insertionLevel{i}})$ time for bucket
  sort and $\mathcal{O}(2^{d\cdot\insertionLevel{i}})$ time for the
  prefix sums, where $2^{d\cdot\insertionLevel{i}}$ is the number of
  cells in the insertion level $\insertionLevel{i}$.  Over all weight
  buckets, the term $|V_i|$ sums up to $|V|$ and Bringmann et
  al.~\cite{bkl-girg-19} show that the same holds for
  $2^{d\cdot\insertionLevel{i}}$.
\end{proof}
\end{backInTime}

\section{Implementation Details}
\label{sec:impl-deta}

\subsection{Avoiding Double Counting Buckets, Cells, and Vertices}
\label{sec:doublecounting}
The algorithm as described in Section~\ref{sec:algorithm} iterates
over pairs of buckets, cells, and vertices.  All three entities need
to be handled correctly to avoid visiting vertex or cell pairs
multiple times.  Consider the recursive Algorithm~\ref{alg:recursive}.
In lines 8 and 9, it is sufficient to consider only cell pairs with
$A\leq B$, because the pairs $(A,B)$ and $(B,A)$ can be handled
simultaneously.  Meaning, if a cell pair is processed by the bucket
pair $(i, j)$, then it needs to be processed by the bucket pair $(j,
i)$ (cf. line 1).  However, the bucket pair $(i, i)$ should occur once
per cell pair.  Alternatively, one can separate the cell pairs $(A,
B)$ and $(B, A)$, but instead consider only bucket pairs $(i, j)$ with
$i\leq j$.  In any case, bucket pairs $(i, i)$ require special
treatment for cell pairs of the form $(A, A)$.  Then, only edges
between vertices $u<v$ should be checked (lines 3,5-6).  If self loops
are desired, the constraint can be relaxed to $u\leq v$.

\subsection{Estimating the Average Degree Parameter}
\label{sec:constants-appendix}

This section covers the estimation for the binomial version of the
model $T>0$.  The calculations for the threshold case $T=0$ are
analogous (and simpler).

Typically, a random graph generator accepts the expected average
degree or the number of edges as an input parameter.  In the following
we describe how binary search can be used to estimate the constant $c$
in the edge probability $p_{uv}$ (Eq.~(\ref{eq:girg-binom})) for a
desired expected average degree.  The constant is found based on the
actual weights instead of their probability distribution, because the
resulting average degree has lower variance and the generator should
be able to accept user-defined weights as well.  Note that we
implement GIRGs without explicitly modeling the constant $c$, because
scaling all weights by $c^T$ emulates the same behavior.

Let $X_{uv}$ be a random indicator variable for the existence of edge $uv$.
\begin{equation*}
\mathbb{E}[X_{uv}] = \mathbb{E}\left[ \min\left\{1, c \cdot \left( \frac{w_uw_v/W}{{||x_u - x_v||}^d} \right)^{1/T} \right\} \right] 
= \mathbb{E}\left[ \min\left\{1, \left( \frac{c^{\frac{T}{d}} \left( \frac{w_uw_v}{W} \right)^{\frac{1}{d}}}{||x_u - x_v||} \right)^{d/T} \right\} \right] 
\end{equation*}
Let $k=c^{\frac{T}{d}} \left( \frac{w_uw_v}{W} \right)^{1/d}$.
To remove the minimum, we distinguish between \emph{short edges} that are guaranteed to exist and \emph{long edges} that exist with probability below 1.
\begin{equation*}
\mathbb{E}[X_{uv}] = Pr(||x_u - x_v|| \leq k) 
+ Pr(k < ||x_u - x_v||) \cdot \mathbb{E}\left[ c\cdot \left( \frac{w_uw_v/W}{{||x_u - x_v||}^d} \right)^{1/T} \mid k < ||x_u - x_v|| \right] \\
\end{equation*}

If $k\geq0.5$ then the weights guarantee the existence of the edge $uv$ independent of position.
For simplicity we assume that $k\leq0.5$ for all vertex pairs.
In the end we subtract an error to account for the ignored pairs.
For any constant $t\leq0.5$, $Pr(||x_u - x_v||\leq t) = (2t)^d$, which is the fraction of the ground space which is covered by a hypercube with radius $t$. 
The probability for a short edge becomes
\begin{equation}
\label{eq:short}
Pr(||x_u - x_v|| \leq k) = (2k)^d = 2^dc^T \left( \frac{w_uw_v}{W} \right)
\end{equation}
The probability density function of $||x_u - x_v||$ between 0 and 0.5 is the derivative of $(2x)^d$, namely $d2^dx^{d-1}$.
We calculate the probability for a long edge based on two specific weights.
\begin{equation}
\label{eq:long}
\begin{aligned}
& Pr(k < ||x_u - x_v||) \cdot \mathbb{E}\left[ c \cdot \left( \frac{w_uw_v/W}{{||x_u - x_v||}^d} \right)^{1/T} \mid k < ||x_u - x_v|| \right] \\
&= Pr(k < ||x_u - x_v||) \cdot \frac{ \int_k^{0.5}c \cdot \left( \frac{w_uw_v/W}{x^d} \right)^{1/T} \cdot d2^dx^{d-1} dx }{Pr(k < ||x_u - x_v||)} \\
&= c \left(\frac{w_uw_v}{W}\right)^{1/T} d 2^d \int_k^{0.5} x^{d-1-d/T} dx \\
&= c \left(\frac{w_uw_v}{W}\right)^{1/T} d 2^d \left[ \frac{1}{d(1-1/T)} \cdot x^{d-d/T} \right]_k^{0.5} \\
&= c \left(\frac{w_uw_v}{W}\right)^{1/T} \frac{d 2^d}{d(1-1/T)} \left( \left(\frac{1}{2}\right)^{d-d/T} - k^{d(1-1/T)} \right) \\
&= c \left(\frac{w_uw_v}{W}\right)^{1/T} \frac{2^d}{1-1/T} \left( \frac{2^{d/T}}{2^d} - \left( c^{\frac{T}{d}} \left( \frac{w_uw_v}{W} \right)^{1/d} \right)^{d(1-1/T)} \right) \\
&= c \left(\frac{w_uw_v}{W}\right)^{1/T} \frac{2^{d/T}}{1-1/T} - c \left(\frac{w_uw_v}{W}\right)^{1/T} \frac{2^d}{1-1/T} c^{T-1} \left( \frac{w_uw_v}{W} \right)^{1-1/T} \\
&= c \left(\frac{w_uw_v}{W}\right)^{1/T} \frac{2^{d/T}}{1-1/T} - c^T \left(\frac{w_uw_v}{W}\right) \frac{2^d}{1-1/T} \\
\end{aligned}
\end{equation}
We add short and long edges (Eq.~\ref{eq:short} and Eq.~\ref{eq:long}).
\begin{equation}
\label{eq:shortlong}
\begin{aligned}
\mathbb{E}[X_{uv}] &=
2^dc^T \left( \frac{w_uw_v}{W} \right) + c \left(\frac{w_uw_v}{W}\right)^{1/T} \frac{2^{d/T}}{1-1/T} - c^T \left(\frac{w_uw_v}{W}\right) \frac{2^d}{1-1/T} \\
&= 2^dc^T \left( \frac{w_uw_v}{W} \right)\left(1+\frac{1}{1/T-1}\right) - c \left(\frac{w_uw_v}{W}\right)^{1/T} \frac{2^{d/T}}{1/T-1} \\ 
&= c^{T}\frac{2^d}{1-T} \left( \frac{w_uw_v}{W} \right) - c \frac{2^{d/T}}{1/T-1} \left(\frac{w_uw_v}{W}\right)^{1/T} \\ 
\end{aligned}
\end{equation}
The expected average degree $\mathbb{E}[\bar d]$ is computed as follows.
\begin{equation}
\label{eq:deg}
\mathbb{E}[\bar{d}] \cdot n = 
c^T\frac{2^d}{1-T} \sum_{u\in V}\sum_{v\neq u} \left( \frac{w_uw_v}{W} \right) - c \frac{2^{d/T}}{1/T-1} \sum_{u\in V}\sum_{v\neq u} \left(\frac{w_uw_v}{W}\right)^{1/T} \\
\end{equation}
We compute the sums for all vertex pairs by subtracting an error for the reflexive edges.
\begin{equation*}
\begin{aligned}
&\sum_{u\in V}\sum_{v\neq u} \left( \frac{w_uw_v}{W} \right) = \sum_{u\in V}\sum_{v\in V}\frac{w_uw_v}{W} - \sum_{v\in V}\frac{w_v^2}{W} = W - \sum_{v\in V}\frac{w_v^2}{W} \\
&\sum_{u\in V}\sum_{v\neq u} \left(\frac{w_uw_v}{W}\right)^{1/T} = \sum_{u\in V} \left( \frac{w_u^{1/T}}{W^{1/T}} \sum_{v\in V} w_v^{1/T} \right) - \sum_{v\in V}\left(\frac{w_v^2}{W}\right)^{1/T}
\end{aligned}
\end{equation*}
We earlier assumed $k\leq0.5$ and still have to subtract an error for the ignored vertex pairs. 
Let $E_R$ be the set of vertex pairs $(u,v)$ with $0.5 < k = c^{\frac{T}{d}} \left( \frac{w_uw_v}{W} \right)^{1/d}$ and let $R$ be the set of vertices with at least one edge in $E_R$.
For each vertex pair in $E_R$, the probability for a short edge is 1 instead of what we got in Equation~\ref{eq:short} and the probability for a long edge is 0 instead of Equation~\ref{eq:long}. 
Therefore, the error due to an edge $(u,v) \in E_R$ can be obtained by subtracting 1 from Equation~\ref{eq:shortlong}.

Now we are ready to find the constant $c$ for a desired average degree using binary search over the monotone function $f(c) = \mathbb{E}[\bar d]$.
The function $f$ is given by Eq.~\ref{eq:deg} substituting the sums and subtracting the error for vertex pairs with $k>0.5$.
\begin{equation*}
\begin{aligned}
f(c) &= 
c^T \cdot \frac{2^d}{n(1-T)} \left( W - \sum_{v\in V}\frac{w_v^2}{W} \right) \\
&- c \cdot \frac{2^{d/T}}{n(1/T-1)} \left(\sum_{u\in V} \left( \frac{w_u^{1/T}}{W^{1/T}} \sum_{v\in V} w_v^{1/T} \right) - \sum_{v\in V}\left(\frac{w_v^2}{W}\right)^{1/T}\right) \\
&- \frac{1}{n} \sum_{(u,v)\in E_R} \left( c^T\frac{2^d}{1-T} \left( \frac{w_uw_v}{W} \right) - c \frac{2^{d/T}}{1/T-1} \left(\frac{w_uw_v}{W}\right)^{1/T} - 1 \right)
\end{aligned}
\end{equation*}
The binary search takes $O(n)$ time to compute various sums and $O(1+|E_R|)$ per evaluation of $f(c)$.
We partially sort the weights for all vertices in $R$ to iterate efficiently over $E_R$.
Since the upper and lower bound for the binary search are found with an exponential search, the size of $R$ might grow until the upper bound is found.
We lazily extend a sorted prefix of the weight array while raising the upper bound.

The time to evaluate $f(c)$ can be quadratically reduced from $O(|E_R|)$ to $O(|R|)$ 
by exploiting that for any two vertices $u$ and $v$ with $w_u \leq w_v$, a pair $(u,x)\in E_R$ implies $(v,x)\in E_R$.
Thus, if we iterate the vertices in $R$ by increasing weight, we can reuse the computations for the last vertex.

\subsection{Efficiently Encoding and Decoding Morton Codes}
\label{sec:morton-appendix}

Recall from Section~\ref{sec:cellaccess} that we linearize the
$d$-dimensional grid of cells using Morton code.  As vertex positions
are given as $d$-dimensional coordinates, we have to convert the
coordinates to Morton codes (i.e., the index in the linearization) and
vice versa.  A Morton code is obtained by bitwise interleaving two or
more coordinates.  E.g., the $2$-dimensional Morton code of the
four-bit coordinates $a=a_3a_2a_1a_0$ and $b=b_3b_2b_1b_0$ is
$a_3b_3a_2b_2a_1b_1a_0b_0$.  Implementation-wise, there are the
following encoding approaches.

\begin{description}
\item[FOR, FOR OPT] Set each bit of the result with shifts and bitwise operations (FOR).
    Since we know the level of a cell, we know the number of relevant bits in each coordinate. 
    Considering only relevant bits improves performance significantly (FOR OPT).
\item[MASKS] For details on this method, we refer to the open-source library \emph{libmorton} \cite{b-l-18} and the authors related blog posts\footnote{\url{https://www.forceflow.be}}.
    The approach is hard to generalize to multiple dimensions.
\item[LUT] A lookup table computed at compile time\footnote{\url{https://github.com/kevinhartman/morton-nd}} can be used.
    The input is divided into chunks; a precomputed result for each chunk is obtained and shifted into place. 
\item[BMI2]
    The \emph{Parallel Bits Deposit/Extract} assembler instructions from Intels Bit Manipulation Instruction Set 2~\cite{intel-manual-19} provide a solution with one assember instruction per input coordinate.
    BMI2 is available on Intel CPUs since 2013 and supported by recent AMD CPUs (Zen).
\end{description}

\begin{figure}
  \centering
  \includegraphics[width=\textwidth]{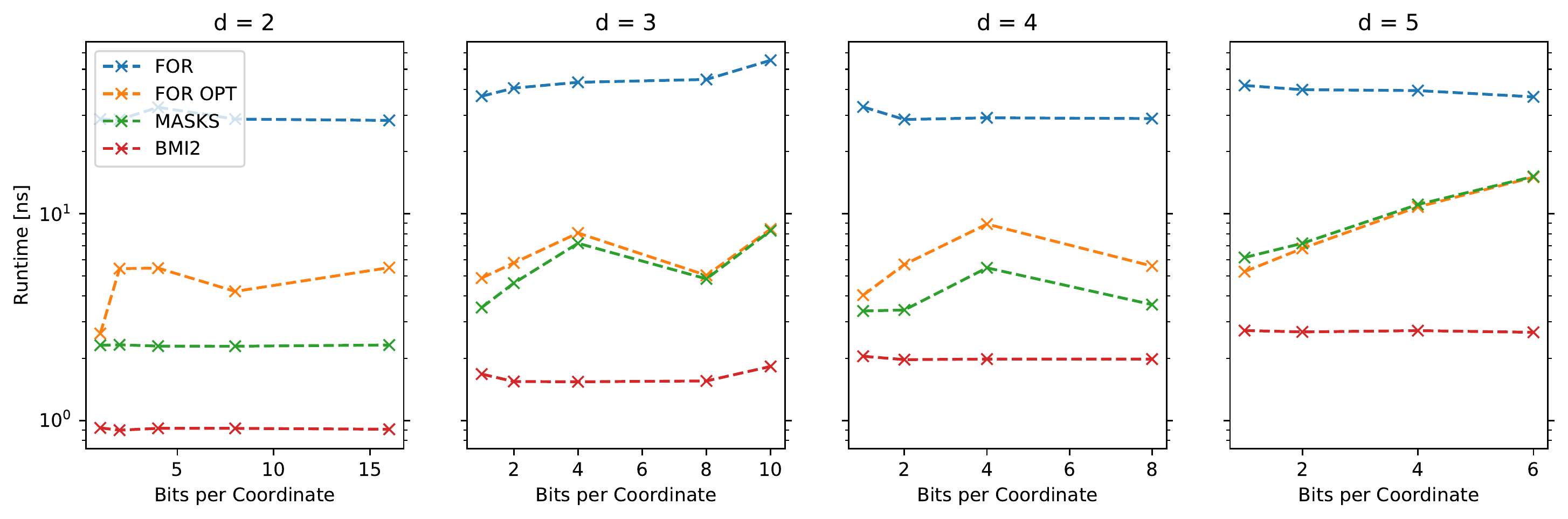}
  \caption{Performance of Morton code generation in dimensions 2 to 5 on an Intel processor. Input coordinates are limited to $\lfloor32/d\rfloor$ bits each, because the result is saved as a 32 bit integer.}
  \label{fig:mortonintel}
\end{figure}

\begin{figure}
	\centering
    \includegraphics[width=\textwidth]{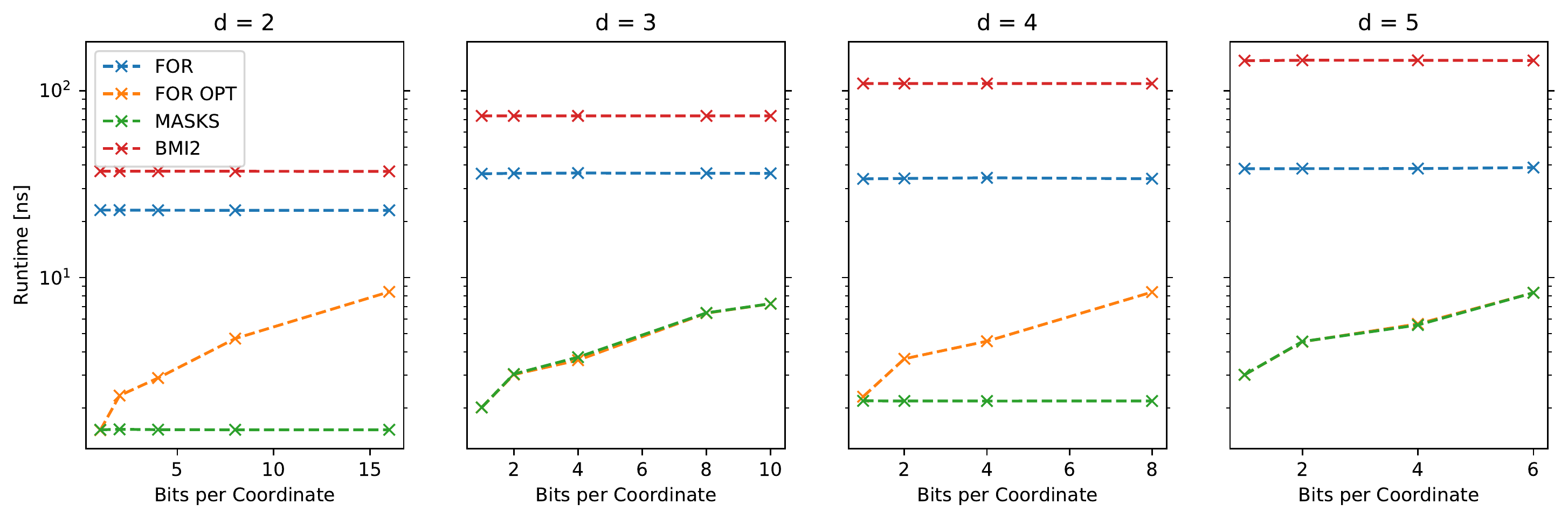}
    \caption{Performance of Morton code generation in dimensions 2 to 5 on an AMD processor. Input coordinates are limited to $\lfloor32/d\rfloor$ bits each, because the result is saved as a 32 bit integer.}
    \label{fig:mortonamd}
\end{figure}

All approaches except LUT support a complementary decoding operation.
We measured the approaches, excluding LUT, on an Intel i7-8550U processor (see Figure~\ref{fig:mortonintel}) and an AMD Ryzen7-2700X (see Figure~\ref{fig:mortonamd}).
On Intel, \textsc{BMI2} is consistently the fastest and at least an order of magnitude faster than \textsc{FOR}.
Surprisingly, \textsc{FOR OPT} is not monotone in the number of bits per coordinate for dimensions below 5. 
Inspection of the generated assembly\footnote{g++8 -std=c++14 -O3 -march=skylake} reveals that the compiler employed SIMD instructions.
On AMD, \textsc{BMI2} is the slowest.
Our GIRG generator uses BMI2 if enabled and the loop with early termination (FOR OPT) otherwise.

\subsection{Avoiding Computationally Expensive Math for HRGs}
\label{sec:nomath-appendix}

The HRG model requires many computationally expensive mathematical
operations.  We significantly improve the performance of the generator
by avoiding or reusing the results of those operations.  The first
optimization applies to the threshold variant and the second
optimization to the binomial version.

For the threshold model, an edge exists if the distance $d$ is smaller than $R$.
Considering how the hyperbolic the distance is defined (Section~\ref{sec:hrgmodel}), 
reformulating it to $\cosh(d) < \cosh(R)$ avoids the expensive $\acosh$, while $\cosh(R)$ remains constant during execution.
Similar to recent threshold HRG generators, we compute intermediate values per vertex such that $\cosh(d)$ can be computed using only multiplication and addition~\cite{flm-cfmdgg-17,p-gprhg-17}.

\begin{figure}
	\centering
	\begin{subfigure}[t]{0.45\textwidth}
		\centering
		\includegraphics[width=\textwidth]{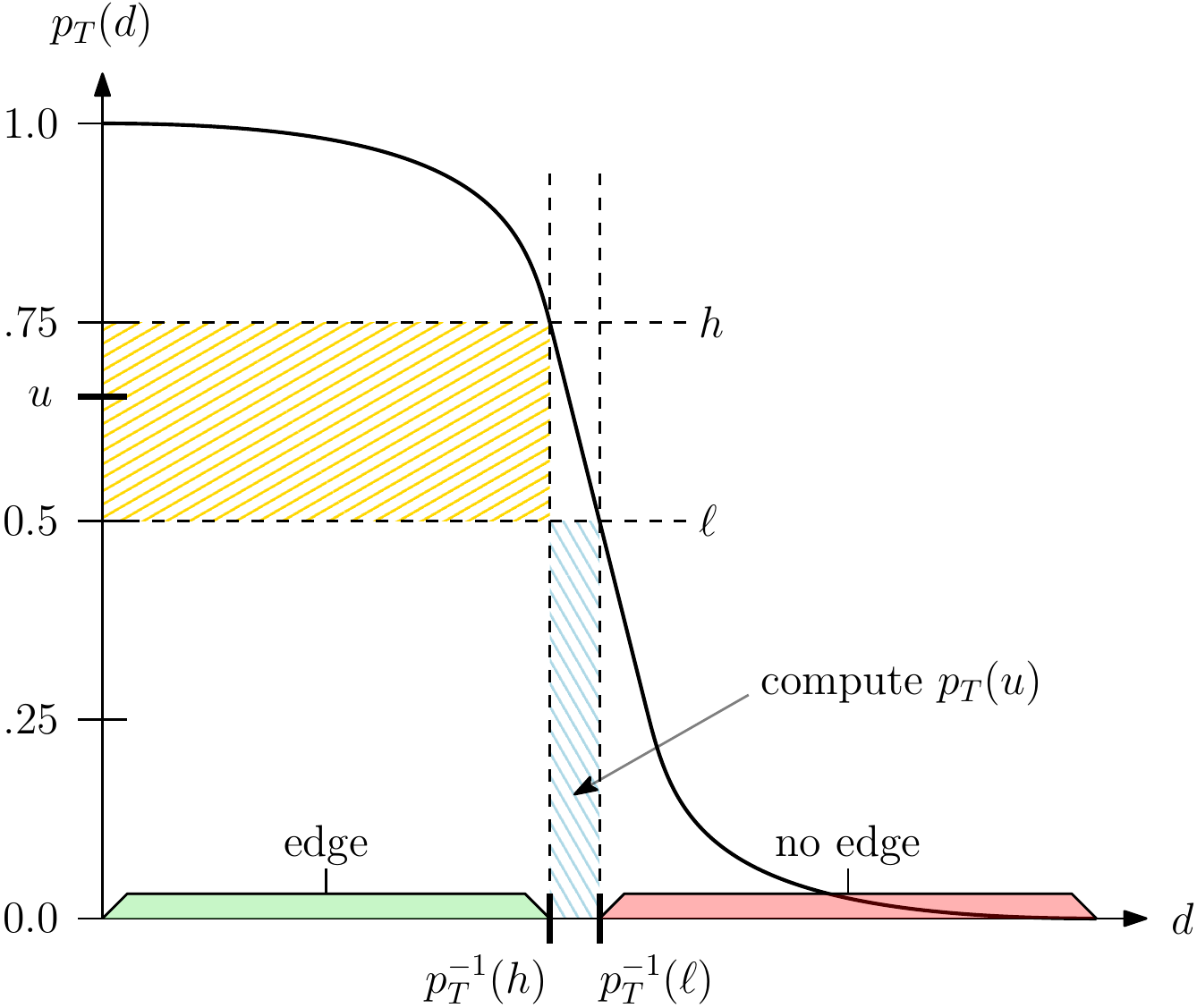}
		\caption{Sketch of the distance filter optimization to avoid computationally expensive mathematical operations, providing a x2 speedup.} 
		\label{fig:filter}
	\end{subfigure}
	\quad
	\begin{subfigure}[t]{0.45\textwidth}
		\centering
		\includegraphics[width=\textwidth]{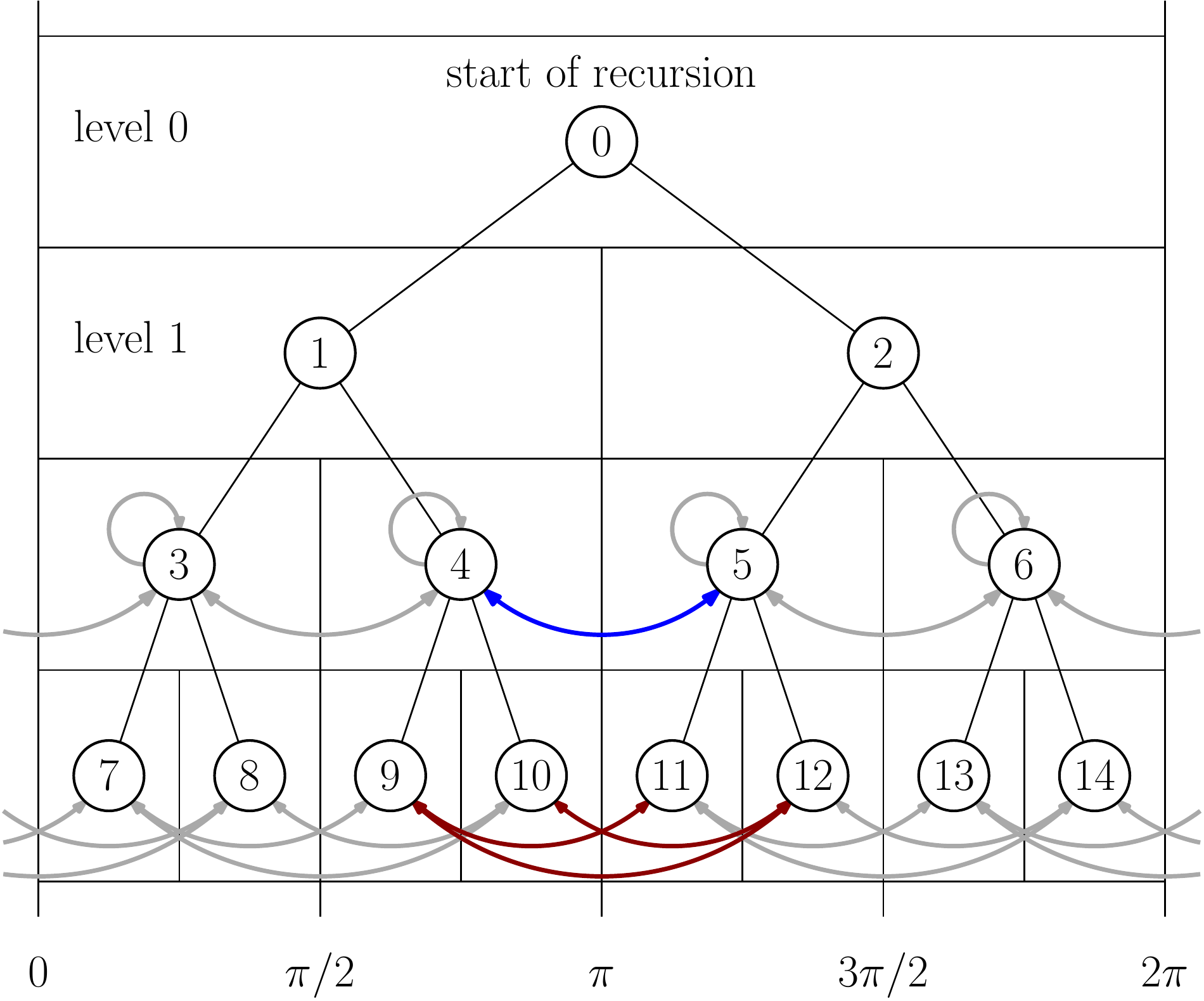}
		\caption{Visited cell pairs up to level 3. The arrows represent 
        the 8 neighboring cell pairs in level 2 and 12 distant cell pairs in level 3.}
		\label{fig:visitcalls}
	\end{subfigure}
	\caption{Distance Filter (left) and tasks for parallelization
          in the $1$-dimensional case (right).}
\end{figure}

For the binomial model, evaluating the connection probability $p_T(d)$
from the optimized $\cosh(d)$ is a performance bottleneck and made up
half of the total run time.
Evaluating $p_T(d)$ includes an expensive exponential function and cannot avoid the $\acosh$ like in the threshold model.
We introduce a distance filter (see Figure~\ref{fig:filter}) to reduce the frequency of the operation resulting in a speedup of approximately factor two.
The process before was: compute the probability $p_T(d)$, sample a uniform random value $u \in [0,1]$, and emit an edge if $u < p_T(d)$.  
The idea of our filter is that we invert the probability function and compute $p_T^{-1}(d)$ in advance for multiple equidistant values between 0 and 1.
Each entry $x \mapsto p_T^{-1}(x)$ in the filter represents the distance --- or rather proximity --- needed for an edge probability of $x$.
During edge sampling, we generate $u \in [0,1]$ before evaluating $p_T$. 
The generated $u$ falls in an interval between two precomputed entries $l \leq u < h$ in our filter.
We know that $p_T$ is monotonically decreasing so $p_T^{-1}(l) \geq p_T^{-1}(u) > p_T^{-1}(h)$, meaning the higher the distance the lower the probability and vice versa.
Instead of emitting an edge if and only if $u < p_T(d)$, we emit an edge if $p_T^{-1}(h) \geq d$ and skip the edge if $p_T^{-1}(l) \leq d$.
Only if $p_T(d)$ is in the interval between $l$ and $h$, the expensive $p_T(d)$ has to be evaluated.
Since $u$ is uniformly distributed, the probability to hit the interval where $p_T(d)$ has to be evaluated is $1/(k-1)$, where $k$ is the number of entries in the filter.
Our generator uses $k=100$.
Additionally, we avoid the $\acosh$ by directly storing $\cosh[p_T^{-1}(x)]$ in the filter.

\subsection{Parallelization}
\label{sec:parallel-appendix}

This section describes how the sampling algorithm can be parallelized
focusing on the preprocessing building the geometric data structure
(Section~\ref{sec:cellaccess}) and on the recursion enumerating pairs
of grid cells for sampling the edges (Section~\ref{sec:loopswap}).
The presented approach applies to the GIRG and HRG implementations.

The preprocessing for a weight bucket $i$ computes the containing cell
on the insertion level for all vertices in $V_i$.  We optimized the
process by processing all weight buckets together.  The containing
cell for all vertices in $V$ is computed in parallel.  We sort
vertices by weight bucket first and by cell second.  Theoretically,
bucket sort results in linear run time.  For the implementation,
however, we use a parallel radix sort instead.  The vertices $V_i$ of
weight bucket $i$ form a contiguous subsequence in $V$.  Moreover,
$V_i$ is sorted by cell, allowing parallel computation of the prefix
sums for all cells in the insertion level of the weight bucket.

The recursion is executed in parallel and experiments suggest a near optimal scaling when the number of threads is a power of two. 
Each thread has a local random generator. 
We use static scheduling to produce deterministic results even for the binomial model.
However, the ordering of edges in the edge list varies, because each thread locally buffers generated edges before writing them while locking a mutex.
We distinguish two stages of execution.  The first stage is to ``saw
off'' the recursion tree at a certain level and collect the omitted
recursive calls as \emph{tasks} to execute in stage two.  A task is
represented by a cell pair from which to pick up the execution later.
One thread collects the tasks by traversing the recursion tree without
sampling any edges (omitting lines 1-6 in
Algorithm~\ref{alg:recursive}).  Meanwhile, the other threads process
the pairs that the main thread passed through.  When all tasks are
collected stage two begins.  In stage two, the threads pick up the
``loose ends'' of the cut recursion tree.  There are three different
types of tasks with varying load.  For $1$-dimensional geometry, level
$\ell>2$, and assuming a number of threads that is constant in $n$,
the types of tasks are the following.
There are $2^\ell$ \emph{heavy tasks} given by a neighboring cell pair
of the form $(A,A)$.  Their number of recursive calls grows
exponentially with each subsequent level implying a load of $O(n)$.
    There are $2^\ell$ \emph{light tasks} given by a neighboring cell pair of the form $(A,A+1)$.
    They produce four recursive calls per subsequent level implying a load of $O(\log n)$.
    Finally, there are $3\cdot2^{\ell-1}$ \emph{constant tasks} given
    by a distant cell pair.  They invoke no recursive calls at all.
    The number of distant cell pairs in a level is explained by
    Figure~\ref{fig:visitcalls}.  For each cell $B$ in level $\ell-1$
    with children $A$ and $A+1$, the distant cell pairs in level
    $\ell$ are $(A,A+2), (A, A+3), (A+1, A+3)$.

    Since heavy tasks dominate the run time during stage two, we
    distribute heavy tasks evenly among all threads.  This is why the
    approach scales best when the number of threads is a power of two.
    The level where we saw off the recursion tree is a tuning
    parameter of the generator.  We choose it, such that there are two
    heavy tasks per thread to reduce load imbalance if one thread
    stalls.  To apply the same scheduling approach to higher
    dimensions it suffices to know that the load of tasks remains
    similar and the number of heavy tasks is~$2^{\ell d}$.


\end{document}